\documentclass[fullpage,twocolumn,letter]{article}
\usepackage{authblk}

\usepackage{cite} 
\usepackage[breaklinks,colorlinks,linkcolor=blue,citecolor=blue,urlcolor=black]{hyperref}
\usepackage{times}
\usepackage{listings}

\usepackage[boxruled]{algorithm2e}
\usepackage{amsthm}
\usepackage{thmtools}
\usepackage{thm-restate}

\long\def\com#1{}

\long\def\xxx#1{}

\long\def\arxiv#1#2{#2}		

\long\def\abbr#1#2{#2}			

\newcommand{\abcite}[2]{\abbr{\cite{#1}}{\cite{#1,#2}}}

\newcommand{\ie}{{\em i.e.}\xspace}
\newcommand{\eg}{{\em e.g.}\xspace}

\newtheorem{definition}{Definition}[section]
\newtheorem{theorem}{Theorem}[section]
\newtheorem{lemma}{Lemma}[section]

\newcommand{\angles}[1]{\langle{#1}\rangle}

\newcommand{\hash}{\ensuremath{\mathsf{Hash}}}

\newcommand{\qsc}{\ensuremath{\mathsf{QSC}}\xspace}
\newcommand{\proposal}[1]{\ensuremath{\langle\mathsf{proposal}\ {#1}\rangle}}
\newcommand{\broadcast}{\ensuremath{\mathsf{Broadcast}}\xspace}
\newcommand{\unicast}{\ensuremath{\mathsf{Unicast}}\xspace}
\newcommand{\choosem}{\ensuremath{\mathsf{ChooseMessage}}\xspace}
\newcommand{\waitm}{\ensuremath{\mathsf{WaitMessage}}\xspace}
\newcommand{\randval}{\ensuremath{\mathsf{RandomValue}}\xspace}

\newcommand{\deliver}{\ensuremath{\mathsf{Deliver}}\xspace}
\newcommand{\receive}{\ensuremath{\mathsf{Receive}}\xspace}
\newcommand{\tlcmsg}[1]{\angles{#1}}
\newcommand{\tlcreq}{\ensuremath{\mathsf{req}}}
\newcommand{\tlcack}{\ensuremath{\mathsf{ack}}}
\newcommand{\tlcwit}{\ensuremath{\mathsf{wit}}}
\newcommand{\tlc}{\ensuremath{\mathsf{TLC}}\xspace}
\newcommand{\tlcr}{\ensuremath{\mathsf{TLCR}}\xspace}
\newcommand{\tlcb}{\ensuremath{\mathsf{TLCB}}\xspace}	
\newcommand{\tlcw}{\ensuremath{\mathsf{TLCW}}\xspace}
\newcommand{\tlcf}{\ensuremath{\mathsf{TLCF}}\xspace}
\newcommand{\tsb}{\ensuremath{\mathsf{TSB}}\xspace}

\newcommand{\kvwrite}{\ensuremath{\mathsf{Write}}\xspace}
\newcommand{\kvread}{\ensuremath{\mathsf{Read}}\xspace}
\newcommand{\kvwait}{\ensuremath{\mathsf{WaitCache}}\xspace}

\newcommand{\qscod}{\ensuremath{\mathsf{QSCOD}}\xspace}

\begin{document}

\title{Que Sera Consensus: \\
	Simple Asynchronous Agreement with \\
	Private Coins and Threshold Logical Clocks}


\author{Bryan Ford}
\affil{Swiss Federal Institute of Technology Lausanne (EPFL)}
\author{Philipp Jovanovic}
\affil{University College London (UCL)}
\author{Ewa Syta}
\affil{Trinity College Hartford}

\maketitle

\begin{abstract}
It is commonly held that asynchronous consensus
is much more complex, difficult, and costly
than partially-synchronous algorithms,
especially without using common coins.
This paper challenges that conventional wisdom
with {\em que sera consensus} (\qsc), an approach to consensus that cleanly decomposes
the agreement problem from that of network asynchrony.
\qsc uses only private coins
and reaches consensus in $O(1)$ expected communication rounds.
It relies on ``lock-step'' synchronous broadcast,
but can run atop a {\em threshold logical clock} (\tlc) algorithm
to time and pace partially-reliable communication
atop an underlying asynchronous network.
This combination
is arguably simpler than partially-synchronous consensus approaches
like (Multi-)Paxos or Raft with leader election,
and is more robust to slow leaders
or targeted network denial-of-service attacks.
The simplest formulations of \qsc atop \tlc
incur expected $O(n^2)$ messages and $O(n^4)$ bits per agreement,
or $O(n^3)$ bits with straightforward optimizations.
An on-demand implementation,
in which clients act as ``natural leaders'' to execute the protocol
atop stateful servers that merely implement passive key-value stores,
can achieve $O(n^2)$ expected communication bits per client-driven agreement.
\end{abstract}


\maketitle

\section{Introduction}
\label{sec:intro}

Most consensus protocols deployed in practice
are derived from Paxos~\cite{lamport98parttime,lamport01paxos},
which relies on leader election and failure detection via timeouts.
Despite decades of refinements and
reformulations~\abcite{lamport01paxos,boichat03deconstructing,ongaro14search,renesse15paxos}{howard15raft,chand19formal},
consensus protocols remain complex,
bug-prone~\cite{abraham17revisiting,leesatapornwongsa16taxdc}
even with formal verification~\cite{chand19formal},
and generally difficult to understand or implement correctly.
Because they rely on network synchrony assumptions for liveness,
their performance is vulnerable to slow leaders
or targeted network denial-of-service
attacks~\cite{clement09making,amir11byzantine}.

Fully-asynchronous consensus algorithms~\cite{rabin83randomized,ben-or83another,bracha85asynchronous,canetti93fast,cachin05random,friedman05simple,moniz06randomized}
address these performance vulnerabilities in principle,
but are even more complex,
often slow and inefficient in other respects,
and rarely implemented in practical systems.
The most practical asynchronous consensus algorithms in particular
rely on common coins~\abcite{rabin83randomized,canetti93fast,cachin02secure,cachin05random,friedman05simple,correia11byzantine,miller16honey,abraham19vaba,syta17scalable}{ben-or85fast,correia06consensus,mostefaoui14signature,duan18beat,abraham19vaba},
which in turn require even-more-complex distributed
setup protocols~\cite{cachin02asynchronous,zhou05apss,kate09distributed,kokoris19bootstrapping}.

This paper makes no attempt to break any complexity-theoretic records,
but instead challenges the conventional wisdom
that fully-asynchronous consensus
is inherently more complex, difficult, or inefficient
than partially-synchronous leader-based approaches.
To this end we introduce {\em que sera consensus} (\qsc),
a randomized consensus algorithm that relies only on private coins
and is expressible in 13 lines of pseudocode
(Algorithm~\ref{alg:qsc} on page~\pageref{alg:qsc}).
This algorithm relies on neither
leader-election nor view-change nor common-coin setup protocols
to be usable in practice.
\qsc does assume private, in-order delivery between pairs of nodes,
but this requirement is trivially satisfied in practice
by communicating over TLS-encrypted TCP connections,
for example~\cite{rfc793,rfc8446}.

\qsc also relies on
a new {\em threshold synchronous broadcast} (\tsb) communication abstraction,
in which coordinating nodes operate logically in lock-step,
but only a subset of nodes' broadcasts in each step may arrive.
\tsb provides each node an operation $\broadcast(m) \rightarrow (R,B)$,
which attempts to broadcast message $m$,
then waits exactly one (logical) time step.
\broadcast then returns a \emph{receive set} $R$ and a {\em broadcast set} $B$,
each consisting solely of messages sent by nodes in the same time step.

The level of reliability a particular \tsb primitive guarantees
is defined by three parameters:
a {\em receive threshold} $t_r$,
a {\em broadcast threshold} $t_b$,
and a {\em spread threshold} $t_s$.
A $\tsb(t_r,t_b,t_s)$ primitive guarantees
that on return from \broadcast on any node,
the returned receive set $R$ contains
the messages sent by at least $t_r$ nodes.
Further, the returned broadcast set $B$ contains
messages broadcast by at least $t_b$ nodes
\emph{and} reliably delivered to
(\ie, appearing in the returned $R$ sets of)
at least $t_s$ nodes in the same time step,
provided the receiving nodes have not (yet) failed during the time step.

To implement this synchronous broadcast abstraction
atop asynchronous networks,
we introduce a class of protocols we call
{\em threshold logical clocks} (\tlc).
Like Lamport clocks~\cite{lamport78time,raynal92about},
\tlc assigns integer numbers to communication events
independently of wall-clock time.
Also like Lamport clocks but unlike
vector clocks~\abcite{fidge88timestamps}{fischer82sacrificing,liskov86highly,mattern89virtual,fidge91logical,raynal92about}
or matrix clocks~\abcite{wuu84efficient,drummond03reducing}{sarin87discarding,ruget94cheaper,raynal92about},
all communicating nodes share a common logical time.
Unlike Lamport clocks,
which merely label arbitrary communication timelines,
\tlc not only labels but also actively {\em paces} communication
so that all nodes progress through logical time in ``lock-step'' --
although different nodes may reach a logical time step
at vastly different real (wall-clock) times.
Nodes that fail (crash) may be conceptually viewed as
reaching some logical time-steps only after an infinite real-time delay.

\tlcr, a simple {\em receive-threshold} logical clock algorithm,
implements $\tsb(t_r,0,0)$ communication for a configurable threshold $t_r$,
in 11 lines of pseudocode
(Algorithm~\ref{alg:tlcr} on page~\pageref{alg:tlcr}).
\tlcb, a {\em broadcast-threshold} logical clock algorithm,
builds on \tlcr to implement $\tsb(t_r,t_b,n)$
\emph{full-spread broadcast} communication
in five lines of pseudocode
(Algorithm~\ref{alg:tlcb} on page~\pageref{alg:tlcb}).
Full-spread broadcast ensures that
at least $t_b$ nodes' messages in each round
reach {\em all} nodes that have not failed by the end of the round,
as required by \qsc.
In a configuration with  $n \ge 3f$ nodes where at most $f$ nodes can fail,
\qsc atop \tlcb (atop \tlcr) ensures that each consensus round
enjoys at least a $1/3$ probability of successful commitment,
yielding three expected consensus rounds per agreement.

This combination represents a complete asynchronous consensus algorithm,
expressible in less than 30 lines of pseudocode total,
and requiring no leader election or common coin setup or other dependencies
apart from standard network protocols like TCP and TLS.
To confirm that the pseudocode representation is not hiding too much complexity,
Appendix~\ref{sec:erl} presents a fully-working model implementation
of \qsc, \tlcb, and \tlcr
in only 37 lines of \href{https://www.erlang.org}{Erlang},
not including test code.

\qsc over \tlcb is usable in $n=2f+1$ configurations
only in the special (but common in practice) case of $f=1$ and $n=3$.
Alleviating this restriction,
\tlcw (Algorithm~\ref{alg:tlcw} on page~\pageref{alg:tlcw})
directly implements $\tsb(t_b,t_b,t_s)$ communication
for configurable $t_b$ and $t_s$,
by proactively confirming the delivery of $t_b$ messages to $t_s$ nodes each,
similar to signed echo broadcast~\cite{reiter94secure}
or witness cosigning~\cite{syta16keeping}
as used in other recent consensus
protocols~\cite{cachin01secure,kokoris16enhancing,abraham19vaba}.
\tlcf (Algorithm~\ref{alg:tlcf} on page~\pageref{alg:tlcf}),
in turn,
implements full-spread $\tsb(t_r,t_b,n)$ communication atop \tlcw
provided $t_r+t_s > n$.
\qsc atop \tlcf (atop \tlcw)
supports minimal $n = 2f+1$ configurations for any $f \ge 0$,
and ensures that each consensus round succeeds
with at least $1/2$ probability,
for two expected consensus rounds per successful agreement.

\qsc incurs only $O(n)$ bits of communication per round
if messages are constant-size.
The \tlc algorithms incur $O(n^4)$ bits per round
if implemented na\"ively,
but this is easily reduced to $O(n^3)$ with simple optimizations.

Further efficiency improvements are feasible with \qscod,
an {\em on-demand} approach to implementing \qsc.
In \qscod, clients wishing to commit transactions
are responsible for driving communication and protocol progress,
and the stateful consensus nodes merely implement passive key-value stores.
\qscod clients effectively serve as ``natural leaders''
to drive communication and consensus efficiency
using only $O(n^2)$ expected bits per client-driven agreement.
In transactional applications,
contention can require some clients to retry
when other clients' proposed transactions ``win.''
The communication costs of retries
may be mitigated using classic techniques such as
by exponential backoff as in CSMA/CD~\cite{ieee802-3},
or by committing batches of gossipped transactions
together in blocks as in	
Bitcoin~\cite{nakamoto08bitcoin}.

In summary, this paper's main contributions are
(a) {\em threshold synchronous broadcast} (\tsb),
a lock-step broadcast communication abstraction
with parameterized delivery thresholds;
(b) {\em que sera consensus} (\qsc),
a simple consensus protocol that builds atop the synchronous \tsb abstraction
but requires neither leader election, view changes, nor common coins; and
(c) {\em threshold logical clocks} (\tlc),
a framework for timing and pacing group communication
that builds lock-step \tsb communication primitives
atop asynchronous underlying networks.

\com{	for PODC submission only
\emph{ Note:
this paper extends and formalizes
ideas proposed informally in an online preprint by the same authors.
An anonymized version of the preprint is available
for reference~\cite{ford19threshold}.
The earlier preprint is not published or under submission.}
}
This paper extends and formalizes
ideas first proposed informally in an earlier preprint
outlining the principles
underlying threshold logical clocks~\cite{ford19threshold}.

\section{Background}
\label{sec:bg}

There are many different formulations of consensus 
and closely-related problems
such as atomic broadcast~\cite{cachin11introduction}.
\qsc's aim is to provide a practical asynchronous consensus protocol
functionally equivalent
to Paxos~\cite{lamport98parttime,lamport01paxos}
or Raft~\cite{ongaro14search}.
In particular,
\qsc provides the equivalent of
Multi-Decree Paxos~\cite{lamport98parttime}
or Multi-Paxos~\cite{chand16formal},
where the goal is to agree on not just one value,
but to commit a {\em sequence} of proposed values progressively
to form a total order.

Because deterministic algorithms
cannot solve asynchronous consensus~\cite{fischer85impossibility},
\qsc relies on randomness for symmetry-breaking~\cite{aspnes03randomized}.
Like Ben-Or's early exponential-time randomized
protocol~\cite{ben-or83another}
but unlike the vast majority of more efficient successors,
\qsc relies only on {\em private randomness}:
coins that each node flips independently of others,
as provided by the random number generators
standard in modern processors and operating systems.
A key goal in particular is not to rely on {\em common coins},
where all nodes choose the same random values.
While protocols based on
secret sharing~\cite{shamir79share,stadler96publicly,schoenmakers99simple}
can produce common coins or public randomness
efficiently~\cite{canetti93fast,cachin05random,syta17scalable},
robust asynchronous {\em setup} of common coins
is essentially as difficult as asynchronous consensus
itself~\cite{cachin02asynchronous,zhou05apss,kate09distributed,kokoris19bootstrapping,ford19threshold}.

\subsection{System model and threat model assumptions}
\label{sec:bg:model}

We assume as usual a group of $n$ nodes communicating over a network
by sending and receiving messages.
A node broadcasts a message to the group
by sending $n$ identical messages,
one to each member including itself.

We assume nodes follow the protocols faithfully as specified.
Nodes can fail, but only by crashing cleanly and permanently,
producing no more messages after the crash.
While it appears readily feasible to extend \qsc and \tlc
to account for Byzantine node behavior~\cite{ford19threshold},
we leave this goal for future work.

We make the standard asynchronous model assumption
that the network {\em eventually} delivers every message,
but only after an arbitrary finite delay of a (network) adversary's choosing.
For simplicity, \qsc also assumes that messages
are delivered in-order between pairs of nodes.
Both assumptions are satisfied in practice if nodes communicate
over TCP~\cite{rfc793} or another reliable, ordered
transport~\cite{rfc908,rfc4960,ford07structured}.
We assume ordered connections never fail unless one endpoint fails:
\eg, timeouts are disabled
and connections are protected against
reset attacks~\cite{luckie15resilience}.

\qsc further assumes either that nodes communicate over private channels
(\eg, encrypted with TLS~\cite{rfc8446}),
or that the network adversary is
{\em content-oblivious}~\cite{aspnes03randomized} or
unable to look into the content of messages or process memory.
Given the prevalence of deep-packet inspection technologies
that intelligent network adversaries can readily employ,
the use of encrypted channels seems safer
than obliviousness assumptions in today's Internet.

\section{Threshold Synchronous Broadcast (\tsb)}
\label{sec:tsb}


Before describing \qsc,
we first introduce a conceptually simple collective communication abstraction
we call {\em threshold synchronous broadcast} (\tsb).
\tsb presumes that a group of $n$ communicating nodes conceptually operates
not in the asynchronous model above
but in lock-step synchronous rounds,
which we will call {\em time steps} or just {\em steps}.
In each step, each node in the group that has not (yet) failed
broadcasts a message to the others,
then receives some subset of all messages sent in that round.
Messages sent are tied to and received {\em only} in the same time-step:
any messages a node does not receive in a given time-step
are simply ``lost'' to that node forever and are never delivered late.

For now we treat threshold broadcast as a primitive API
that a (slightly unrealistic) underlying network might conceivably provide.
We will later develop algorithms to implement this abstraction
atop asynchronous networks.

A \tsb primitive does not in general offer perfect communication reliability.
\tsb instead guarantees reliability
only so as to meet certain threshold parameters,
hence the name.
We say that a broadcast primitive provides
$\tsb(t_r,t_b,t_s)$ reliability if:
(a) it guarantees that each node receives the messages broadcast by
at least $t_r$ nodes in the same time-step, and
(b) it guarantees that the messages sent by at least $t_b$ nodes
are each reliably delivered or {\em spread} to at least $t_s$ nodes each.
A perfectly-reliable \tsb primitive would be $\tsb(n,n,n)$,
guaranteeing that every message sent in each round reaches every node.
A completely-unreliable \tsb primitive would be $\tsb(0,0,0)$,
which makes no message delivery guarantees at all
and hence might not be very useful,
although it {\em might} sometimes deliver some messages.

For simplicity, we assume time is measured in integer units,
as if each broadcast were a ``real-time'' operation
taking exactly one unit of time.
That is, each node broadcasts exactly one message at time-step $0$
intended to be received at time-step $1$,
at step $1$ each node broadcasts
exactly one message to be received at step $2$,
and so on.

\com{
Threshold broadcast does not guarantee that a particular broadcast message
will actually be received by all nodes in the group, however --
or by anyone, for that matter.
It guarantees only that a {\em threshold number} of nodes' messages
sent any time-step $s$ are delivered to every node
that reaches step $s+1$ without failing.
}

We represent the threshold broadcast primitive as a single API function,
$\broadcast(m) \rightarrow (R,B)$.
When a node $i$ calls $\broadcast(m)$ at time-step $s$,
the network broadcasts message $m$,
waits exactly one time-step,
then returns two message sets $R$ and $B$ to the caller.
Returned set $R$ is a set of messages
that node $i$ received during time-step $s$.
Returned set $B$ indicates a set of messages
that were each broadcast reliably
to at least $t_s$ nodes each.

When a message is reliably broadcast to a node $j$,
this means that node $j$ will receive the message in the same time-step $s$,
{\em unless} $j$ fails before time-step $s$ completes.
Thus, a failed node $j$ may count toward
towards a broadcast message's spread threshold $t_s$,
provided we can be certain
that node $j$ {\em would} receive the message had it not failed.\footnote{
An alternative, perhaps mathematically cleaner conception of node failures
is to presume that {\em all} nodes always ``eventually''
reach {\em all} time steps.
But when a node $j$ ``fails'' before step $s$,
this simply means that $j$ reaches $s$ after an infinite real-time delay,
\ie, $j$ reaches $s$ at wall-clock time $\infty$.
Adopting this viewpoint,
\tsb's promise that a message $m \in B$
will ``eventually'' reach at least $t_s$ nodes becomes unconditional,
independent of node failure.
This is because any failed node $j$ in that set
conceptually {\em does} reach step $s$ and receive $m$,
only at real-time $\infty$.
}

A $\tsb(t_r,t_b,t_s)$ primitive guarantees on return from \broadcast
that $R$ contains the messages broadcast by at least $t_r$ nodes in step $s$,
and that $B$ contains messages broadcast by at least $t_b$ nodes,
each of which is reliably delivered to at least $t_s$ nodes during step $s$.
\tsb makes no other delivery guarantees, however.
For example, \tsb makes no guarantee even that $i$'s own message $m$
is within the sets $R$ or $B$ returned to $i$.
Further, two nodes $i$ and $j$
may see different received sets $R_i \ne R_j$
and/or different broadcast sets $B_i \ne B_j$
returned from their respective \broadcast calls
in the same step.
And two messages $m_1 \in B_i$ and $m_2 \in B_i$,
both returned from the same node $i$'s \broadcast call,
may have been broadcast to {\em different} node sets
$N_1$ and $N_2$ respectively:
\tsb guarantees only that $|N_1| \ge t_s$ and $|N_2| \ge t_s$
and not that $N_1 = N_2$.

\begin{definition}
A network offers a $\tsb(t_r,t_b,t_s)$ primitive provided:
\begin{itemize}
\item	Lock-step synchrony:
	A call to $\broadcast(m)$ at any integer time-step $s$
	completes and returns at time-step $s+1$,
	unless the node fails before reaching time-step $s+1$.
\item	Receive threshold:
	If a node $i$'s call to $\broadcast(m)$ at step $s$ returns $(R,B)$,
	then there is a node set $N_R \subseteq \{1,\dots,n\}$
	such that $|N_R| \ge t_r$,
	and $R$ contains exactly the set of messages $m_j$
	broadcast by nodes $j \in N_R$ during step $s$.
\item	Broadcast threshold:
	If a node $i$'s call to $\broadcast(m)$ at step $s$ returns $(R,B)$,
	then there is a node set $N_B \subseteq \{1,\dots,n\}$
	such that $|N_B| \ge t_b$,
	and $B$ contains exactly the set of messages $m_j$
	broadcast by nodes $j \in N_B$ during step $s$.
\item	Spread threshold:
	If a node $i$ called $\broadcast(m)$ at step $s$,
	which returned $(R,B)$ such that a message $m' \in B$,
	then there are at least $t_s$ nodes
	whose message sets $R$ to be returned from \broadcast in step $s$
	will include message $m'$.
\end{itemize}
\label{def:tsb}
\end{definition}

The special case of $\tsb(t_r,t_b,n)$, where $t_s=n$,
represents a particularly-useful {\em full-spread} broadcast primitive.
Such a primitive guarantees that in each time-step $s$,
each reliably-broadcast message returned in any node's $B$ set
is delivered to all $n$ nodes during time-step $s$,
apart from any nodes that fail before step $s$ completes.
A full-spread \tsb has the useful property
that the $B$ set returned to any node
is a subset of the $R$ set returned to any (other) node.

\begin{lemma}
In a network of $n$ nodes offering a full-spread $\tsb(t_r,t_b,n)$ primitive,
if a node $i$'s call to $\broadcast(m_1)$ at step $s$
returns $(R_1,B_1)$,
and a node $j$'s call to $\broadcast(m_2)$ at the same step $s$
returns $(R_2,B_2)$,
then $B_1 \subseteq R_2$.
\end{lemma}
\begin{proof}
This property directly follows from the broadcast spread property.
\end{proof}

\section{Que Sera Consensus}
\label{sec:qsc}

In this section we describe the 
{\em que sera consensus} (\qsc) protocol,
then analyze its correctness and complexity.

\subsection{Building consensus atop \tsb}
\label{sec:consensus}

First we will define more precisely what properties we desire
from a consensus protocol built atop a \tsb communication abstraction.
We will focus on implementing a multi-consensus protocol,
functionally analogous  
to Multi-Paxos~\cite{lamport98parttime,chand16formal},
where nodes agree on a sequence of values (a log)
instead of just one value as in (single-decree) Paxos~\cite{lamport98parttime}.
For simplicity,
the rest of this paper refers to multi-consensus simply as consensus.

We represent a consensus protocol $\mathcal{P}$ as a process
that runs concurrently and indefinitely (or until it fails) on each of a set of $n$ nodes
communicating via threshold reliable broadcast.
We consider $\mathcal{P}$ to be an algorithm parameterized by four functions:
\choosem, \deliver, \randval, and \broadcast.
The first two function parameters represent
an ``upcall-style'' interface
to the application or higher-level protocol.
$\mathcal{P}$ invokes \choosem
to ask the application to choose the next message
that the application wishes to commit.
$\mathcal{P}$ invokes \deliver
to deliver committed
 messages up to the application.
$\mathcal{P}$'s remaining two function parameters represent
its lower-level interface to the network and operating system.
$\mathcal{P}$ calls \randval to choose a numeric value
using node-private randomness,
and $\mathcal{P}$ calls \broadcast
to broadcast a message using the \tsb primitive
and obtain the broadcast's results one time-step later.

For simplicity of presentation and reasoning,
our formulation of consensus protocols will deliver not just individual messages
but entire {\em histories}, ordered lists
cumulatively representing all messages committed and delivered so far.
An easy and efficient standard practice
is to represent a history as the typically constant-size head
of a tamper-evident log~\cite{schneier99secure,crosby09efficient} or blockchain~\cite{nakamoto08bitcoin},
each log entry containing a hash-link to its predecessor.
Thus, the fact that histories conceptually grow without bound
is not a significant practical concern.

Intuitively,
the key properties we want from $\mathcal{P}$ are liveness, validity, and consistency.
Liveness means that $\mathcal{P}$ regularly keeps advancing time
and delivering progressively-longer histories via \deliver,
forever or until the node fails. 
Validity means that any message delivered by any node
is one that the application recently returned via \choosem
on some (potentially different) node.
Finally, consistency means that
a history delivered in an earlier time-step
is a prefix of any history delivered in a later time-step,
both on the same node and across distinct nodes.

\begin{definition}\label{def:abroadcast}
A multi-consensus protocol $\mathcal{P}$
is a potentially-randomized algorithm that takes function parameters
$(\choosem,\allowbreak \deliver,\allowbreak \randval,\allowbreak \broadcast)$
and behaves as follows:
\begin{itemize}
\item	Liveness:
	If $h$ is the longest history $\mathcal{P}$ has delivered
	by time-step $s$ on some non-failing node $i$,
	or $h = []$ if $\mathcal{P}$ has not yet invoked \deliver
	by step $s$,
	then there is some future time-step $s' > s$
	at which $\mathcal{P}$ invokes $\deliver(h')$
	with some history $h'$ strictly longer than $h$
	(\ie, $|h'| > |h|$).
\item	Validity:
	For some constant $\delta \ge 0$,
	if $\mathcal{P}$ invokes $\deliver(h'\ ||\ [p])$
	at time-step $s'$ on node $j$,
	then $p$ is a proposal $\proposal{i,m,r}$
	that node $i$ returned from an invocation of \choosem
	at some time-step $s \le s'$,
	where $s' - s \le \delta$.
\item	Consistency:
	if $\mathcal{P}$ invokes $\deliver(h)$ at time-step $s$ on node $i$,
	and $\mathcal{P}$ later invokes $\deliver(h')$ at time-step $s' \ge s$
	on node $j$ (either the same or a different node),
	then $h$ is a prefix of $h'$.
\end{itemize}

$\mathcal{P}$'s behavior above is contingent
on its function parameters satisfying operational specifications
described below. 
\end{definition}

The application upcall function parameters \choosem and \deliver
may behave in arbitrary application-specific fashions,
provided they do not interfere with the operation of \qsc
or the lower layers it depends on
(\eg, by corrupting memory,
or de-synchronizing the nodes via unexpected calls to \broadcast).
\choosem always returns some message,
which may be an empty message if the application has nothing useful
to broadcast in a given time-step.
The \randval function must return a numeric value (integer or real)
from a nontrivial random distribution,
containing at least two values
each occurring with nonzero probability,
and from the same random distribution on every node.
The \broadcast function must correctly implement
a threshold reliable broadcast primitive
as described above
in Section~\ref{sec:tsb}.

\subsection{Que Sera Consensus (\qsc) algorithm}
\label{sec:qsc:alg}

\begin{algorithm*}[t!]
\SetCommentSty{textnormal}
\DontPrintSemicolon
\SetKwInOut{Input}{Input}
\SetKwFor{Forever}{forever}{}{end}

\Input{configuration parameters $n,t_r,t_b,t_s$, where
$t_r > 0$, $t_b > 0$, and $t_s = n$}
\Input{function parameters \choosem, \deliver, \randval, \broadcast}
Run the following concurrently on each communicating process $i \in \{1, \dots, n\}$:\;
$h \leftarrow []$	\tcp*{consensus history is initially empty}
\Forever(\tcp*[f]{loop forever over consensus rounds}){}{
	$m \leftarrow$ \choosem() \tcp*{choose some message for node $i$ to propose}
	$r \leftarrow \randval()$ \tcp*{choose proposal priority using private randomness}
	$h' \leftarrow h\ ||\ [\proposal{i,m,r}]$
		\tcp*{append one new proposal to our view of history}
	$(R',B') \leftarrow \broadcast(h')$
		\tcp*{broadcast our proposal and advance one time-step}
	$h'' \leftarrow$ any best history in $B'$
		\tcp*{choose best eligible history we know of so far}
	$(R'',B'') \leftarrow \broadcast(h'')$
		\tcp*{re-broadcast the best proposal we see so far}
	$h \leftarrow$ any best history in $R''$
		\tcp*{choose our best history for next consensus round}
	\If(\tcp*[f]{history $h$ has no possible competition}){$h \in B''$ and $h$ is uniquely best in $R'$}{
		\deliver($h$)
			\tcp*{deliver newly-committed history}
	}
}
\caption{Que Sera Consensus (\qsc)}
\label{alg:qsc}
\end{algorithm*}
Algorithm~\ref{alg:qsc} concisely summarizes
que sera consensus (\qsc),
a simple multi-consensus algorithm
satisfying the above specification.
The \qsc algorithm is a process
that runs on each of the $n$ nodes forever or until the node fails.
Each iteration of the main loop implements a single {\em consensus round},
which may or may not deliver a new history $h$.

The \qsc algorithm depends on \broadcast providing
a full-spread $\tsb(t_r,t_b,n)$ abstraction (Section~\ref{sec:tsb}).
Each consensus round invokes
this \broadcast primitive twice,
thus taking exactly two broadcast time-steps per consensus round.

Each node maintains its own view of history, denoted by $h$,
which increases in size by one entry per round.
Each node does not build strictly on its own prior history
in each round, however,
but can discard its own prior history in favor of adopting one
built by another node.
In this way \qsc's behavior is analogous to Bitcoin~\cite{nakamoto08bitcoin},
in which the ``longest chain'' rule may cause a miner
to abandon its own prior chain in favor of a longer chain on a competing fork.
\qsc replaces Bitcoin's ``longest chain'' rule with a ``highest priority'' rule,
however.

At the start of a round,
each node $i$ invokes \choosem
to choose an arbitrary message $m$ to propose in this round,
possibly empty if $i$ has nothing to commit.
This message typically represents a transaction or block of transactions
that the application running on $i$ wishes to commit,
on behalf of itself or clients it is serving.
Node $i$ also uses \randval to choose a random numeric {\em priority} $r$
using node-private (not shared) randomness.
Based on this information,
node $i$ then appends a new proposal $\proposal{i,m,r}$
to the prior round's history
and broadcasts this new proposed history $h'$ using \broadcast, which
returns sets $R'$ and $B'$.

From the set $B'$ of messages that \tsb promises
were reliably broadcast to all non-failed nodes
in this first \broadcast call,
node $i$ picks any history $h''$ (not necessarily unique)
having priority at least as high as any other history in $B'$,
and broadcasts $h''$.
This second \broadcast call returns history sets $B''$ and $R''$ in turn.
Finally, $i$ picks from the resulting set $R''$
any {\em best} history (Definition~\ref{def:qsc:besth}), \ie, 
any history (again not necessarily unique)
with priority at least as high as any other in $R''$,
as the resulting history for this round and
the initial history for the next round
from node $i$'s perspective.
We define the priority of a nonempty history
as the priority of the last proposal it contains.
Thus, a history $h = [\dots, \proposal{i,m,r}]$
has priority $r$.

\begin{definition}\label{def:qsc:besth}
A history $h$ is {\em best} in a set $H$ if $h \in H$
and no history $h' \in H$
has priority strictly greater than $h$.
\end{definition}

The resulting history each node arrives at in a round
may be either {\em tentative} or {\em final}.
Each node decides separately
whether to consider its history tentative or final,
and nodes may make different decisions on finality in the same round.
Each node $i$ then {\em delivers} the resulting history $h$
to the application built atop \qsc,
via a call to \deliver,
only if node $i$ determined $h$ to be final.
If $i$ decides that $h$ is tentative,
it simply proceeds to the next round,
leaving its view of history effectively undecided
until some future round eventually delivers a final history.

A node $i$ decides that its resulting history $h$ in a round is {\em final} 
if (a) $h$ is in the set $B''$ returned from the second broadcast,
and (b) $h$ is the {\em uniquely best} history in the set $R'$
returned from the first broadcast.

\begin{definition}\label{def:qsc:uniqueh}
A history $h$ is {\em uniquely best} in a set $H$ if $h \in H$
and there is no other history $h' \ne h$
such that $h'$ is also in $H$
and has priority greater than or equal to that of $h$.
\end{definition}

\com{
\begin{definition}\label{def:qsc:finalh}
Given histories $h, h', h''$ and $h'''$, $(R', B') \leftarrow \broadcast(h')$
$(R'', B'') \leftarrow \broadcast(h'')$, where $h''$ is best in $B''$,
we say that $h$ is {\em final} if $h \in B''$ and $h$ is uniquely best in $R'$, \ie, 
$\not \exists h''' \in R'$ with a strictly higher priority than $h$.
\end{definition}
}

This pair of finality conditions is sufficient to ensure
that {\em all} nodes will have chosen exactly the same resulting history $h$
at the end of this round,
as explained below --
even if other nodes may not necessarily realize
that they have agreed on the same history.
Since all future consensus rounds
must invariably build on this common history $h$
regardless of which nodes' proposals ``win'' those future rounds,
node $i$ can safely consider $h$ final and deliver it to the application,
knowing that all other nodes will similarly build on $h$
regardless of their individual finality decisions.

\subsection{Correctness of \qsc}

While the \qsc algorithm itself is simple,
analyzing the correctness of any consensus protocol involves some subtleties,
which we examine first intuitively then formally.
The main challenges are first,
ensuring that the histories it delivers are consistent,
and second,
that it determines rounds to be final and delivers longer histories
``reasonably often.''
This section only states key lemmas and the main theorem;
the proofs may be found in Appendix~\ref{sec:proofs:qsc}.

\subsubsection{Safety}
Consistency is \qsc's main safety property.
We wish to ensure that if at some step $s$ a node $i$ delivers history $h$,
and at some later step $s' \geq s$ any node $j$ delivers history $h'$,
then $h$ is a prefix of $h'$.
That is, every node consistently builds on any history prefix
delivered earlier by any other node.
To accomplish this,
in \qsc each node $i$ delivers a history $h$ only if $i$ can determine
that all other non-failed nodes must also become aware
that $h$ exists and can be chosen in the round,
and that no other node could choose any {\em other} history in the round.

Each node first chooses some best (highest-priority) eligible history $h''$
from the set $B'$
returned by the first \broadcast call.
Set $B'$ includes only {\em confirmed} histories:
those that $i$ can be certain all non-failed nodes
will become aware of during the round.
By the full-spread requirement ($t_s=n$)
on the \tsb primitive,
history $h''$ must be included in the $R'$ sets
returned on all non-failed nodes,
ensuring this awareness requirement
even if other nodes choose different histories.

After the second $\broadcast(h'')$ call, 
all histories returned in $R''$ and $B''$ are confirmed histories.
Further,
any history $h \in B''$ is not just confirmed but {\em reconfirmed},
meaning that all non-failed nodes will learn during the round
not just that $h$ exists but also
{\em the fact that $h$ was confirmed}.
Each node $i$ chooses, as its tentative history
to build on in the next round,
some (not necessarily unique) highest-priority confirmed proposal
among those in $R''$ that $i$ learns about.

Finally, $i$ considers $h$ committed
and actually delivers it to the application
only if $h$ is both reconfirmed (in $B'')$ and {\em uniquely} best
within the broader set $R'$ that includes all confirmed proposals
other nodes {\em could} choose in this round.
These two finality conditions ensure that
(a) all nodes know that $h$ is confirmed
and thus {\em can} choose it as their tentative history, and
(b) all nodes {\em must} choose $h$
because it is uniquely best among all the choices they have.
This does not guarantee that other nodes will {\em know} that $h$ is committed,
however:
another node $j \ne i$ might not observe the finality conditions,
but will nevertheless ``obliviously'' choose and build on $h$,
learning only in some later round that $h$ is final.

A key first step is showing that
any node's (tentative or final) history at any round
builds on {\em some} node's (tentative or final) history
at any prior round.

\begin{restatable}{lem}{lempreservation}
History preservation:
If a consensus round starting at time-step $s$
has initial history $h_{s_i}$ on node $i$,
then at any earlier consensus round starting at step $s' < s$,
there exists some node $j$ whose initial history $h_{s'_j}$ in that round
is a strict prefix of $h_{s_i}$.
\label{lem:preservation}
\end{restatable}

Consistency, \qsc's main safety property,
relies on the fact that each node $i$
delivers a resulting history in a consensus round
only when $i$ is sure that {\em all} nodes
will choose the same resulting history in that round.
The above lemma in turn guarantees that the histories
all nodes build on and potentially deliver in the future
must build on this common history.

\begin{restatable}{lem}{lemagreement}
Agreement on delivery:
If \qsc delivers history $h_{(s+2)_i}$ on node $i$
at the end of a consensus round starting at time-step $s$,
then the resulting history $h_{(s+2)_j}$ of every node $j$ in the same round
is identical to $h_{(s+2)_i}$.
\label{lem:agreement}
\end{restatable}

\com{	Redundant with the informal discussion above. -BF

Informally, a proposal $h$ is only delivered if two finality conditions are met:
{\em i)} $h \in B''$ and {\em ii)} $h$ is uniquely best in $R'$, that is, there is no other history 
in $R'$ with greater or equal priority. 
If $\qsc$ delivers $h$ for $i$, then the first condition ensures $i$'s delivered history
must be among those that any other node $j$ has to choose from as $h$ is included 
in every node $j$'s set $R''_j$.
The second finality condition ensures that no competing history exists that $j$ could choose
in preference to $i$'s delivered history as $h$ must be uniquely best in $R'$ and all nodes
are aware of all histories in $R'$. 
Note that neither of these conditions guarantee
that $j$ is symmetrically aware that agreement has occurred:
$j$ may in fact be unaware of this fact,
but will nevertheless ``obliviously'' pick and build on $i$'s delivered history,
learning only in some later round that this history prefix is final.
}

\subsubsection{Liveness}
\label{sec:qsc:liveness}

The other main prerequisite to \qsc's correctness
is ensuring its liveness, \ie, that it makes progress.
Unlike safety, liveness is probabilistic:
\qsc guarantees only that each node has a ``reasonable'' nonzero chance of 
delivering a committed history in each round.
This ensures in turn that for each node $i$, after any time-step $s$,
with probability 1 there exists some future time-step $s' \ge s$
at which $i$ delivers some (next) final history.

\qsc's liveness depends on the network scheduling message delivery
independently of the contents of proposals.
More precisely,
\qsc assumes that the network underlying \tsb primitive
chooses the sets $N_R$ and $N_B$,
determining which messages each node receives
and learns were reliably broadcast
(Definition~\ref{def:tsb}),
independently of the random priority values contained in the proposals.
As mentioned before, in practice, we can satisfy this assumption
either by assuming
a content-oblivious network scheduler~\cite{aspnes03randomized},
or by using private channels
(\eg, encrypted by TLS~\cite{rfc8446}).

\begin{restatable}{lem}{lemliveness}
If the network delivery schedule is independent of proposal priorities
and $p_t$ is the probability that two nodes tie for highest priority,
then each node delivers a history in each round independently
with probability at least ${t_b}/{n} - p_t$.
\label{lem:liveness}
\end{restatable}

\subsubsection{Overall correctness of \qsc}

The above lemmas in combination ensure
that \qsc correctly implements consensus.

\begin{restatable}{thm}{thmqsc}
\qsc implements multi-consensus on $n$ nodes (Definition~\ref{def:abroadcast})
atop a full-spread \tsb primitive $\tsb(t_r,t_b,n)$
where $t_r > 0$ and $t_b > 0$.
\label{thm:qsc}
\end{restatable}

Notice that \qsc's correctness theorem makes no direct assumptions
about the number of failing nodes,
in particular not mentioning the standard majority requirement $n > 2f$.
This is because the number of failing nodes affects only liveness,
and \qsc depends for liveness
on the undelying \tsb primitive's unconditional promise
to return from each \broadcast call in exactly one time-step.
It will prove impossible to implement a $\tsb(t_r,t_b,t_s)$ primitive
where $t_r > n-f$ or $t_b > n-f$,
because the primitive would have to collect messages from failed nodes
in order to return $R$ and $B$ sets of the promised size after each step.
But that is \tsb's problem, not \qsc's.

\com{
\xxx{
Randomness assumptions:
could in principle be any nontrivial distribution
on integer or real numbers,
provided all nodes choose priorities from the same distribution.
The distribution could even a single binary coin-flip (fair or unfair),
although this would be inefficient
as the consensus success probability would decrease exponentially in $n$.
For efficiency, we want the probability of choosing any particular priority
to be $O(1/n)$,
so that the probability of the highest-priority proposal tying with any other
is constant,
giving each consensus round a constant success probability.
}
}

\subsection{Asymptotic complexity of \qsc}
\label{sec:qsc:complexity}

\com{
\xxx{
Each node's history $h$ grows by one proposal
every two threshold broadcast time-steps.

$R''$ is a subset of $R'$.

Because $t_b > n/2$,
this implies that $t_b > 0$
and in turn that all sets $B$ and $R$
returned by calls to \broadcast are nonempty.

Any eligible history is included in the $R'$ returned to all nodes.

Any node's $R''$ is a subset of the eligible histories for that round.

Any history appearing in the $B''$ returned by any node
is universally known to be eligible at the end of the round.

Message complexity to agree on and deliver a next history:
$O(n)$ expected,
because each of the $n$ nodes transmits one message per round,
and provided $t_b/n$ is constant,
it takes $O(1)$ expected rounds between successful deliveries.
}}

Implementing \qsc na\"ively,
the histories broadcast in each round would grow linearly with time.
We can make \qsc efficient, however,
by adopting the standard practice of representing histories
as tamper-evident logs
or blockchains~\cite{schneier99secure,nakamoto08bitcoin,crosby09efficient}.
Each broadcast needs to contain only
the latest proposal or {\em head} of the history,
which refers to its predecessor
(and transitively to the entire history)
via a cryptographic hash.
Since \qsc does not need anything but the head in each round,
this is sufficient.
With this optimization,
the two messages each node broadcasts in each round are $O(1)$ size.
The total message and communication complexity of \qsc
is therefore $O(n^2)$ per round across the $n$ nodes,
assuming each broadcast requires $n$ unicast transmissions
(efficient broadcast would eliminate a factor of $n$).
Provided $t_b/n$ is constant,
it takes a constant expected number of rounds (namely $n/t_b$)
to commit and deliver a new consensus result,
so each consensus progress event likewise incurs
$O(n^2)$ expected communication complexity.
This analysis neglects the cost of implementing
the underlying \tsb abstraction that \qsc builds on, of course,
an issue we address later.

\section{Threshold Logical Clocks}
\label{sec:tlc}

Since the \tsb abstraction seems somewhat tailor-made
for implementing consensus,
it would not be particularly useful
if it were almost as difficult to implement \tsb
as to implement consensus directly.
Fortunately, there are multiple clean and simple ways
to implement the \tsb primitive atop realistic, fully-asynchronous networks.

For this purpose we develop several variants of
a lower-level abstraction we call {\em threshold logical clocks} (\tlc).
The main purpose of \tlc is to provide the illusion of lock-step synchrony
that the \tsb abstraction presents and that \qsc relies on,
despite the underlying network being asynchronous.
Secondarily, a \tlc also conveniently provides the communication patterns needed
to implement the threshold reliability that the \tsb abstraction promises.


\com{
\tlc is inspired in part by Lamport clocks~\cite{lamport78time,raynal92about},
vector clocks~\cite{fischer82sacrificing,liskov86highly,mattern89virtual,fidge91logical,raynal92about},
and matrix clocks~\cite{wuu84efficient,sarin87discarding,ruget94cheaper,drummond03reducing,raynal92about}.
While these classic notions of virtual time label an unconstrained event history
to enable before/after comparisons, \tlc in contrast labels {\em and} constrains
events to ensure that a threshold of nodes in a group progress through logical
time in a quasi-synchronous lock-step fashion. In particular, a \tlc node
reaches time step $s+1$ only after a threshold of all participants has
reached time $s$ and a suitable threshold amount of round-trip communication
has demonstrably occurred since then. Put differently, \tlc instead constrains
nodes so that they must coordinate with a threshold of nodes in order to ``earn
the privilege'' of creating a new event and incrementing their notion of the
logical time.
}

The rest of this section is organized as follows: In Section~\ref{sec:tlcr} we
introduce \tlcr, a simple {\em receive-threshold logical clock} algorithm
realizing $\tsb(t_r, 0, 0)$. Afterwards, in Section~\ref{sec:tlcb}, we discuss
\tlcb a {\em broadcast-threshold logical clock} algorithm building on top of
\tlcr to provide full-spread broadcast communication $\tsb(t_r, t_b, n)$.
In Section~\ref{sec:tlcw} we then present \tlcw,
a {\em witnessed-threshold logical clock} algorithm
implementing $\tsb(t_r, t_b, t_s)$ communication,
amending some of the restrictions of \tlcb.
Finally, in Section~\ref{sec:tlcf},
we describe \tlcf,
which builds full-spread witness broadcast communication
$\tsb(t_r, t_b, n)$ on top of $\tlcr$ and $\tlcw$.
Proofs for theorems in this section
are in Appendix~\ref{sec:proofs:tlc}.

\com{
We start in Section~\ref{sec:tlcr} with \tlcr a basic \tlc algorithm
that provides a $\tsb(t_r,0,0)$ primitive,
which implements the lock-step synchrony illusion
and guarantees that each node receives a threshold number of messages at
each logical time-step,
but provides no broadcast reliability guarantees
and hence is by itself insufficient to implement \qsc.

\xxx{
The simplest, {\em unwitnessed} flavor of \tlc yields the simplest protocol,
and guarantees using simple counting arguments
a broadcast threshold $t_b > 0$,
provided that the total number of nodes $n$ is sufficiently greater
than the maximum number $f$ that might fail.
This flavor of \tlc achieves the consensus lower bound of $n > 2f$
only for the special case of $f=1$ and $n=3$, however,
requiring $n$ to grow faster than the lower bound if $f>1$.

The second, {\em witnessed} flavor of \tlc requires
slightly greater complexity and more communication per logical time-step,
but is still practical,
and cleanly supports the lower bound of $n > 2f$.
}

Finally, we present in Section~\ref{sec:tlcf} full-spread threshold synchronous broadcast
implementation $\tsb(t_r, t_b, n)$ on top of \tlcw .
}

\subsection{\tlcr: receive-threshold synchrony on asynchronous networks}
\label{sec:tlcr}

\begin{algorithm*}[t!]
\caption{$\tlcr(m)$, using a threshold logical clock
	to implement receive-threshold synchronous broadcast}
\SetCommentSty{textnormal}
\DontPrintSemicolon
\SetKwInOut{Config}{Configuration}
\SetKwInOut{State}{Persistent state}
\SetKwInOut{Input}{Function input}
\SetKwInOut{Output}{Function output}
\SetKwFor{Forever}{forever}{}{end}

\Config{node number $i$, number of nodes $n$,
	receive threshold $t_r \le n$}
\Config{functions \receive, \broadcast representing
	underlying asynchronous network API}
\State{receive message-set log $\vec{R}$,
	initialized to the singleton list $[\{\}]$}
\BlankLine
\Input{message $m$ to broadcast in this time-step}
\Output{sets $(R,B)$ of messages received in this time-step, and reliably broadcast (always empty)}
\BlankLine
$\vec{R} \leftarrow \vec{R}\ ||\ [\{\}]$
	\tcp*{start a new logical time-step with an empty receive message-set}
$\broadcast(\tlcmsg{i,m,|\vec{R}|,\vec{R}_{|\vec{R}|-1}})$
	\tcp*{broadcast our message, current time-step, and last message-set}
\While(\tcp*[f]{loop until we reach receive threshold $t_r$
		to advance logical time}){$|\vec{R}_{|\vec{R}|}| < t_r$}{
	$\tlcmsg{j,m',s',R'} \leftarrow \receive()$
		\tcp*{await and receive next message $m'$ from any node $j$}
	\uIf(\tcp*[f]{message $m'$ was sent in
			our current time-step}){$s' = |\vec{R}|$}{
		$\vec{R}_{|\vec{R}|} \leftarrow \vec{R}_{|\vec{R}|} \cup
				\{\tlcmsg{j,m'}\}$
			\tcp*{collect messages received in this time-step}
	}
	\ElseIf(\tcp*[f]{message $m'$ is from the next step due to
			in-order channels}){$s' > |\vec{R}|$}{
		$\vec{R}_{|\vec{R}|} \leftarrow \vec{R}_{|\vec{R}|} \cup R'$
			\tcp*{virally adopt message-set $R'$ that $j$ used
				to advance history}
	}
}
\KwRet{$(\{m'\ |\ \tlcmsg{j,m'} \in \vec{R}_{|\vec{R}|}\}, \{\})$}
	\tcp*{return the received message set and an empty broadcast set}
\label{alg:tlcr}
\end{algorithm*}

Algorithm~\ref{alg:tlcr} implements a $\tsb(t_r,0,0)$ abstraction
atop an asynchronous network,
ensuring that each node receives messages from at least $t_r$ nodes
during each logical time-step.
Although \tlcr tolerates messages being scheduled and delayed arbitrarily,
it makes the standard assumption that a message broadcast by any node
is {\em eventually} delivered to every other non-failing node.
For simplicity, \tlcr also assumes
messages are delivered {\em in-order} between any pair of nodes,
\eg, via any sequenced point-to-point transport such as TCP.

In \tlcr, each node broadcasts a message
at the beginning of each step $s$,
then waits to receive at least $t_r$ messages from step $s$.
\tlcr internally logs the receive-set it returns from each step in $\vec{R}$,
whose length tracks the current time-step.

Each node's broadcast in each step also includes the receive-set
with which it completed the previous step.
If a node receives any message from step $s+1$
before collecting a threshold of messages from $s$,
it immediately completes step $s$
using the previous receive-set it just obtained.
Because of the above assumption messages are pairwise-ordered (\eg, by TCP),
a node never receives a message for step $s+2$ or later
before receiving a message for step $s+1$ from the same node,
and thus never needs to ``catch up'' more than one step at a time.

As an alternative to including
the previous step's receive set in each broadcast,
\tlcr could simply defer the processing of messages for future steps
until the receive threshold is met for the current step.
This approach eliminates the pairwise ordered-delivery assumption,
at the potential cost of slightly slower progress in practice
because messages arriving early from future time steps
cannot "virally" help delayed nodes make progress.

\begin{restatable}{thm}{thmtlcr}
\tlcr (Algorithm~\ref{alg:tlcr})
implements a $\tsb(t_r,0,0)$ communication primitive
with receive threshold $0 \le t_r \le n$,
provided at most $f \le n - t_r$ nodes fail.
\label{thm:tlcr}
\end{restatable}

\subsubsection{Asymptotic complexity of \tlcr}
\label{sec:tlcr:complexity}

Since each node broadcasts exactly one message per time-step,
\tlcr incurs a total messsage complexity of $O(n^2)$ per round
across the $n$ nodes,
assuming each broadcast requires $n$ unicasts.

If the messages passed to \tlcr are constant size,
then \tlcr incurs a communication complexity of $O(n^3)$ per round
because of \tlcr's inclusion of the previous round's receive-set
in each broadcast.
Implementing \tlcr na\"ively,
if the messages passed to \tlcr are $O(n)$ size,
then total communication complexity is therefore $O(n^4)$ per round,
and so on.

A simple way to reduce this communication cost, however,
is simply to defer the processing of messages for future time steps
that arrive early,
as discussed above.
This way, broadcasts need not include the prior round's receive-set,
so communication complexity is only $O(n^2)$ per round
when application messages are constant size.

Another approach is to replace the application messages themselves
with constant-size references (\eg, cryptographic hashes)
to out-of-line blocks,
and to use a classic IHAVE/SENDME protocol as in USENET~\cite{rfc1036}
on the point-to-point links between nodes
to transmit only messages that the receiver has not yet obtained
from another source.
In brief, on each point-to-point message transmission
the sender first transmits the summary message containing only references;
the sender then waits for the receiver to indicate for which references 
the receiver does not yet have the corresponding content;
and finally the sender transmits only the content
of the requested references.
With this standard practice in gossip protocols,
each node typically receives each content block only once
(unless the node simultaneously downloads the same block from multiple sources
to minimize latency at the cost of bandwidth).

Implementing \tlcr in this way,
each round incurs a communication complexity of $O(n^3)$ per round
even if the messages passed to \tlcr are $O(n)$ size.
This is because each node proposes only one new message $m$ per round
and each node receives its content only once,
even if the prior round receive-set in each round's proposal
{\em refers to} $O(n)$ messages from the prior round
via constant-size references (\eg, hashes).

\subsection{\tlcb: broadcast-threshold synchrony atop \tlcr}
\label{sec:tlcb}

\begin{algorithm*}[t!]
\caption{$\tlcb(m)$, using two \tlcr steps to implement
	spread-threshold synchronous broadcast}
\SetCommentSty{textnormal}
\DontPrintSemicolon
\SetKwInOut{Config}{Configuration}
\SetKwInOut{State}{Persistent state}
\SetKwInOut{Input}{Function input}
\SetKwInOut{Output}{Function output}
\SetKwFor{Forever}{forever}{}{end}

\Config{node number $i$, number of nodes $n$, receive threshold $t_r \le n$,
	and spread threshold $t_s \le t_r$}
\Config{function \tlcr implementing $\tsb(t_r,0,0)$
	unreliable synchronous broadcast}
\BlankLine
\Input{message $m$ to broadcast in this $\tlcb$ step}
\Output{sets $(R,B)$ of messages received and reliably broadcast, respectively,
	in this $\tlcb$ step}
\BlankLine
$(R',\_) \leftarrow \tlcr(m)$
	\tcp*{broadcast our message unreliably in the first \tlcr step}
$(R'',\_) \leftarrow \tlcr(R')$
	\tcp*{rebroadcast all messages we received in the first \tlcr step}
$R \leftarrow \bigcup(\{R'\} \cup R'')$
	\tcp*{collect all messages we received directly or indirectly}
$B \leftarrow \{m'\ |$ at least $t_s$ message-sets in $R''$ contain $m'\}$
	\tcp*{messages we know were seen by at least $t_s$ nodes}
\KwRet{$(R,B)$}
\label{alg:tlcb}
\end{algorithm*}

Although \tlcr provides no broadcast threshold guarantees,
in suitable network configurations,
\tlcb (Algorithm~\ref{alg:tlcb}) does so by
simply using two successive \tlcr rounds
per (\tlcb) time-step.
In brief,
\tlcb uses its second \tlcr invocation
to broadcast and gather information about 
which messages sent in the first \tlcr invocation
were received by enough ($t_s$) nodes.
Simple ``pigeonhole principle'' counting arguments ensure that
enough ($t_b$) such first-round messages are so identified,
provided the configuration parameters
$n$, $t_r$, $t_b$, and $t_s$ satisfy certain constraints.
These constraints are specified in the following theorem,
whose detailed underlying reasoning may be found
in Appendix~\ref{sec:proofs:tlc}.

\com{
Let $f_r = n-t_r$ be the number of {\em received failures} allowed
by a \tsb primitive
(nodes' messages {\em not} received in a given step).
Similarly,
let $f_b = n-t_b$ be the number of {\em broadcast failures} allowed
(messages sent but {\em not} broadcast reliably to at least $t_s$ nodes),
and let $f_s = n-t_s$ be the number of {\em spread failures} allowed
(nodes failing to receive a reliably-broadcast message).
}

\begin{restatable}{thm}{thmtlcbpartial}
If 	$0 < t_r \le n-f$,
	$0 < t_s \le t_r$,
	$0 < t_b \le n-f_b$ where $f_b = t_r (n-t_r) / (t_r-t_s+1)$,
and at most $f$ nodes fail,
then \tlcb (Algorithm~\ref{alg:tlcb})
implements a $\tsb(t_r,t_b,t_s)$ partial-spread broadcast abstraction.
\label{thm:tlcb-partial}
\end{restatable}

Suppose we desire a configuration tolerating up to $f$ node failures,
and we set $n = 3f$, $t_r = 2f$, $t_b = f$, and $t_s = f+1$.
Then $f_b = 2f(3f-2f)/(2f-(f+1)+1) = 2f$,
so $t_b \le n-f_b$ as required.
This \tlcb configuration therefore reliably delivers
at least $t_b = n/3$ nodes' messages
to at least $t_s = n/3+1$ nodes each in every step.

\subsubsection{Full-spread reliable broadcast using \tlcb}
\label{sec:tlcb:full}

If we configure \tlcb above
to satisfy the additional constraint that $t_r + t_s > n$, 
then it actually implements {\em full-spread reliable broadcast} or
$\tsb(t_r,t_b,n)$.
This constraint reduces to the classic {\em majority} rule, $n > 2f$,
in the case $t_r = t_s = f$ where at most $f$ nodes fail.

Under this constraint,
each of the (at least $t_b$) messages
in the set $B$ returned from \tlcb on any node
is guaranteed to appear in the set $R$
returned from the same \tlcb round on {\em every} node
that has not yet failed at that point.
Intuitively, this is because the first \tlcr call
propagates each message in $B$ to at least $t_s$ nodes,
every node collects $R'$ sets from at least $t_r$ nodes
during the second \tlcr call,
and since $t_r + t_s > n$ these spread and receive sets must overlap.

\begin{restatable}{thm}{thmtlcbfull}
If	$0 < t_r \le n-f$,
	$0 < t_s \le t_r$,
	$t_r + t_s > n$,
	$0 < t_b \le n-f_b$ where $f_b = t_r (n-t_r) / (t_r-t_s+1)$,
and at most $f$ nodes fail,
then \tlcb (Algorithm~\ref{alg:tlcb})
implements a $\tsb(t_r,t_b,n)$ full-spread broadcast abstraction.
\label{thm:tlcb-full}
\end{restatable}

Under these configuration constraints, therefore,
\tlcb provides a \tsb abstraction
sufficient to support \qsc (Section~\ref{sec:qsc}).
This consensus algorithm supports the optimal $2f+1$ node count
for the special case of $f=1$ and $n=3$,
which in practice is an extremely common and important configuration.
\com{PJ: Could we add some more context here, it's not obvious why 
\tlcb only works for this parameterization?
	tried to above. -BF}
For larger $f$, however, running \qsc on \tlcb
requires $n$ to grow faster than $2f+1$.
This limitation motivates {\em witnessed} \tlc, described next,
which is slightly more complex
but allows \qsc to support an optimal $n=2f+1$ configuration for any $f \ge 0$.


\subsubsection{Asymptotic complexity of \tlcb}
\label{sec:tlcb:complexity}

Implementing \tlcb na\"ively on na\"ively-implemented \tlcr
yields a total communication complexity of $O(n^4)$ per round
if the messages passed to \tlcb are of size $O(1)$.

As discussed above in Section~\ref{sec:tlcr:complexity}, however,
this cost may be reduced by delaying the processing of messages
for future time steps,
or by using hash-references
and an IHAVE/SENDME protocol on the point-to-point links.
In this case, \tlcr incurs a communication complexity of $O(n^3)$ per round
with $O(1)$-size messages,
because the set $R'$ in the second broadcast
is not actually an $O(n)$-length list of $O(n)$-size messages,
but is rather an $O(n)$-size list of $O(1)$-size hash-references to messages
whose content each node receives only once.

\subsection{\tlcw: witnessed threshold logical clocks}
\label{sec:tlcw}

\begin{algorithm*}[t!]
\caption{$\tlcw(m)$, a witnessed threshold logical clock
	implementing $\tsb(t_b,t_b,t_s)$ synchronous broadcast}
\SetCommentSty{textnormal}
\DontPrintSemicolon
\SetKwInOut{Config}{Configuration}
\SetKwInOut{State}{Persistent state}
\SetKwInOut{Input}{Function input}
\SetKwInOut{Output}{Function output}
\SetKwFor{Forever}{forever}{}{end}

\Config{node number $i$, number of nodes $n$,
	broadcast threshold $t_b \le n$,
	spread threshold $t_s \le n$}
\Config{functions \receive, \broadcast, \unicast representing
	underlying asynchronous network API}
\State{message receive log $\vec{R}$ and broadcast log $\vec{B}$,
	each initialized to a singleton list $[\{\}]$}
\BlankLine
\Input{message $m$ to broadcast in this time step}
\Output{sets $(R,B)$ of messages received, and reliably broadcast,
	in this time-step}
\BlankLine
$(\vec{R},\vec{B}) \leftarrow (\vec{R}\ ||\ [\{\}], \vec{B}\ ||\ [\{\}])$
	\tcp*{start a new logical time-step with empty receive and broadcast sets}
$N_A \leftarrow \{\}$
	\tcp*{initially empty witness acknowledgment set for our message $m$}
$\broadcast(\tlcmsg{\tlcreq,i,m,|\vec{R}|,\vec{R}_{|\vec{R}|-1},\vec{B}_{|\vec{R}|-1}})$
	\tcp*{broadcast our request, current time-step, and last message-sets}
\While(\tcp*[f]{loop until we reach broadcast threshold $t_b$ to advance logical time}){$|\vec{B}_{|\vec{R}|}| < t_b$}{
	\Switch(\tcp*[f]{receive the next message from any node}){$\receive()$}{
		\uCase(\tcp*[f]{request message $m'$ from node $j$
				in the same time-step}){$\tlcmsg
					{\tlcreq,j,m',|\vec{R}|,\_,\_}$}{
			$\vec{R}_{|\vec{R}|} \leftarrow \vec{R}_{|\vec{R}|}
					\cup \{\tlcmsg{j,m'}\}$
				\tcp*{collect messages we received and witnessed in this time-step}
			$\unicast(j, \tlcmsg{\tlcack,i,m',|\vec{R}|,
						\vec{R}_{|\vec{R}|-1},
						\vec{B}_{|\vec{R}|-1}})$
				\tcp*{acknowledge node $j$'s request
					as a witness}
		}
		\uCase(\tcp*[f]{acknowledgment of our request $m$
				from node $j$}){$\tlcmsg
					{\tlcack,j,m,|\vec{R}|,\_,\_}$}{
			$N_A \leftarrow N_A \cup \{j\}$
				\tcp*{collect acknowledgments of our request message}
			\If(\tcp*[f]{our message has satisfied the spread
					threshold $t_s$}){$|N_A| = t_s$}{
				$\broadcast(\tlcmsg{\tlcwit,i,m,|\vec{R}|,
						\vec{R}_{|\vec{R}|-1},
						\vec{B}_{|\vec{R}|-1}})$
					\tcp*{announce our message $m$ as
						fully witnessed}
			}
		}
		\uCase(\tcp*[f]{announcement that $j$'s message $m'$
				was witnessed by $t_s$ nodes}){$\tlcmsg
					{\tlcwit,j,m',|\vec{R}|,\_,\_}$}{
			$\vec{B}_{|\vec{R}|} \leftarrow \vec{B}_{|\vec{R}|}
					\cup \{\tlcmsg{j,m'}\}$
				\tcp*{collect fully-witnessed messages received in this time-step}
		}
		\Case(\tcp*[f]{message $m'$ is from the next step due to
				in-order channels}){$\tlcmsg
					{\_,\_,\_,|\vec{R}|+1,R',B'}$}{
			$(\vec{R}_{|\vec{R}|},\vec{B}_{|\vec{R}|}) \leftarrow
				(\vec{R}_{|\vec{R}|} \cup R',
				\vec{B}_{|\vec{R}|} \cup B')$
				\tcp*{virally adopt message-sets
					that $j$ used to advance history}
		}
	}
}
\KwRet{$(\{m'\ |\ \tlcmsg{j,m'} \in \vec{R}_{|\vec{R}|}\},
	 \{m'\ |\ \tlcmsg{j,m'} \in \vec{B}_{|\vec{R}|}\})$}
	\tcp*{return the final message sets for this time-step}
\label{alg:tlcw}
\end{algorithm*}

\tlcw (Algorithm~\ref{alg:tlcw})
in essence extends \tlcr (Section~\ref{sec:tlcr})
so that each node $i$ works {\em proactively} in each round
to ensure that at least $t_b$ nodes' messages
are received by at least $t_s$ nodes each,
and waits until $i$ can confirm this fact
before advancing to the next logical time-step.

\tlcw accomplishes this goal by having each node
run an echo broadcast protocol~\cite{cachin11introduction}
in parallel,
to confirm that its own message has been received by at least $t_s$ nodes,
before its message is considered {\em threshold witnessed}
and hence ``counts'' toward a goal of $t_b$ such messages.
Variants of this technique have been used in other recent consensus protocols
such as ByzCoin~\cite{kokoris16enhancing} and VABA~\cite{abraham19vaba}.
As in \tlcr,
a slow node can also catch up to another node at a later timestep
by reusing the set of threshold-witnessed messages
that the latter node already used to advance logical time.

\begin{restatable}{thm}{thmtlcw}
If	$0 < t_b \le n-f$,
	$0 < t_s \le n-f$,
and at most $f$ nodes fail,
then \tlcw (Algorithm~\ref{alg:tlcw})
  implements a $\tsb(t_b,t_b,t_s)$ partial-spread broadcast abstraction.
\label{thm:tlcw}
\end{restatable}

\subsection{\tlcf: full-spread threshold synchronous broadcast with \tlcw and \tlcr}
\label{sec:tlcf}

\begin{algorithm*}[t!]
\caption{$\tlcf(m)$, using \tlcw and \tlcr to implement
	full-spread threshold synchronous broadcast}
\SetCommentSty{textnormal}
\DontPrintSemicolon
\SetKwInOut{Config}{Configuration}
\SetKwInOut{State}{Persistent state}
\SetKwInOut{Input}{Function input}
\SetKwInOut{Output}{Function output}
\SetKwFor{Forever}{forever}{}{end}

\Config{node number $i$, number of nodes $n$,
	receive threshold $t_r \le n$,
	broadcast threshold $t_b \le n$,
	and spread threshold $t_s \le n$,
	such that $t_r+t_s > n$}
\Config{function \tlcr implementing $\tsb(t_r,0,0)$
	receive-threshold synchronous communication}
\Config{function \tlcw implementing $\tsb(t_b,t_b,t_s)$
	witnessed broadcast-threshold synchronous communication}
\BlankLine
\Input{message $m$ to broadcast in this $\tlcf$ step}
\Output{sets $(R,B)$ of messages reliably broadcast and received, respectively,
	in this $\tlcf$ step}
\BlankLine
$(R',B) \leftarrow \tlcw(m)$
	\tcp*{broadcast at least $t_b$ messages to at least $t_s$ nodes each}
$(R'',\_) \leftarrow \tlcr(R')$
	\tcp*{rebroadcast all messages we received in the first step}
$R \leftarrow \bigcup(\{R'\} \cup R'')$
	\tcp*{collect all messages we received directly or indirectly}
\KwRet{$(R,B)$}
\label{alg:tlcf}
\end{algorithm*}

While \tlcw directly implements
only partial-spread threshold synchronous broadcast,
similar to \tlcb above
we can ``bootstrap'' it to full-spread synchronous broadcast
in configurations satisfying $t_r + t_s > n$.
\tlcf, shown in Algorithm~\ref{alg:tlcf},
simply follows a \tlcw round with a \tlcr round.
By exactly the same logic as in \tlcb,
this ensures
that each message in the broadcast set $B$ returned from \tlcw
propagates to {\em every} node that has not failed
by the end of the subsequent \tlcr round,
because all the broadcast-spread sets in \tlcw
overlap with all the receive-sets in the subsequent \tlcr.

\begin{restatable}{thm}{thmtlcf}
If	$0 < t_r \le n-f$,
	$0 < t_b \le n-f$,
	$0 < t_s \le n-f$,
	$t_r + t_s > n$,
and at most $f$ nodes fail,
then \tlcf (Algorithm~\ref{alg:tlcf})
implements a $\tsb(t_r,t_b,n)$ full-spread broadcast abstraction.
\label{thm:tlcf}
\end{restatable}

\section{On-demand client-driven implementation of \tlc and \qsc}
\label{sec:od}

In appropriate network configurations,
\qsc (Section~\ref{sec:qsc:alg})
may be implemented atop
either full-spread \tlcb (Section~\ref{sec:tlcb:full})
or \tlcf (Section~\ref{sec:tlcf}).
Supporting a fully-asynchronous underlying network,
these combinations progress and commit consensus decisions continuously
as quickly as network connectivity permits.
Using the optimizations described in
Sections~\ref{sec:qsc:complexity},
\ref{sec:tlcr:complexity},
and~\ref{sec:tlcb:complexity},	
these implementations incur expected total communication costs
of $O(n^3)$ bits per successful consensus decision and \qsc history delivery.

We would like to address two remaining efficiency challenges, however.
First, in many practical situations
we want consensus to happen not continuously but only {\em on demand},
leaving the network idle and consuming no bandwidth
when there is no work to be done (\ie, no transactions to commit).
Second, it would be nice if \qsc could achieve the optimal lower bound
of $O(n^2)$ communication complexity, at least in common-case scenarios.
\xxx{ fill in reference to lower bound }

\subsection{Consensus over key-value stores}

With certain caveats,
we can achieve both efficiency goals above by implementing \qsc and \tlc
in a client-driven architecture.
In this instantiation,
the $n$ stateful nodes representing the actual consensus group members
are merely passive servers that implement only
a locally-atomic {\em write-once} key-value store.

\begin{definition}
A write-once store serves \kvwrite and \kvread requests from clients.
A $\kvwrite(K, V)$ operation atomically writes value $V$
under key $K$ provided no value exists yet in the store under key $K$,
and otherwise does nothing.
A $\kvread(K) \rightarrow V$ operation 
returns the the value written under key $K$,
or empty if no value has been written yet under key $K$.
\end{definition}

In practice the $n$ servers can be any of innumerable distributed key/value stores
supporting locally-atomic writes~\cite{kafka,nats,rabbitmq,redis}. 
The $n$ servers might even be standard Unix/POSIX file systems,
mounted on clients via NFS~\cite{rfc7530} for example.\footnote{
A standard way to implement \kvwrite atomically on a POSIX file system
is first to write the contents of $V$ to a temporary file
(ensuring that a partially-written file never exists under name $K$),
attempt to hard-link the temporary file to a filename for the target name $K$
(the POSIX \texttt{link} operation fails if the target name already exists),
and finally unlink the temporary filename
(which deletes the file if the \texttt{link} operation failed).
}

\subsection{Representing the \qsc/\tlc state machine}

Each of the $n$ key/value stores
implicitly represents the current state
of that node's \qsc/\tlc state machine,
as simulated by the clients.
One transition in each server's state machine
is represented by exactly one atomic key/value \kvwrite.
All keys ever used on a node inhabit a well-defined total order
across both \tlc time-steps and state transitions within each step.
To track the $n$ servers' consensus states and drive them forward,
each client locally runs $n$ concurrent threads or processes,
each simulating the state machine of one of the servers.

To read past consensus history from each server $i$
and catch up to its current state,
a client's local state-machine simulation thread $i$
simply reads keys from $i$ in their well-defined sequence,
replaying the \qsc/\tlc state machine defined by their values at each transition,
until the client encounters the first key not yet written on the server.
Clients that are freshly started or long out-of-date
can catch up more efficiently
using optimizations discussed later in Section~\ref{sec:od:opt}.

To advance consensus state,
each client's $n$ simulation threads coordinate locally
to decide on and (attempt to) write new key-value pairs to the servers,
representing nondeterministic but valid state transitions on those servers.
Each of these writes may succeed or fail due to races
with other clients' write attempts.
In either case, the client advances its local simulation
of a given server's state machine
\emph{only} after a read to the appropriate key,
\ie, according to the state transition defined by whichever client
won the race to write that key.

We outline only the general technique here.
\qscod, Algorithm~\ref{alg:qscod} in Appendix~\ref{sec:odb},
presents pseudocode for
a specific example of a client-driven on-demand implementation
of \qsc over \tlcb.

\subsection{Complexity analysis}

Implementing \qsc over \tlcb in this way in \qscod,
a client that is already caught up to the servers' states
incurs $O(n^2)$ expected communication bits
to propose and reach agreement on a transaction.
This is because the client reads and writes only 
a constant number of $O(n)$-size messages to the $O(n)$ servers
per consensus round,
and \qsc requires a constant expected number of rounds
to reach agreement.

The client eliminates the need for broadcasts
by effectively serving in a ``natural leader'' role,
analogous to an elected leader in Paxos --
but without Paxos's practical risk
of multiple leaders interfering with each other
to halt progress entirely.
When multiple \qscod clients ``race''
to drive the servers' implicit \qsc/\tlc state machines
concurrently in one consensus round,
the round still progresses normally and completes
with the same (constant) success probability
as with only one client.

Under such a time of contention,
only one's client's proposal can ``win'' and be committed in the round,
of course.
Since in transactional applications
clients whose proposals did not win may need to retry,
this contention can increase communication costs and server load,
even though the system is making progress.
Since each consensus round is essentially a shared-access medium,
one simple way to mitigate the costs of contention
is using random exponential backoff,
as in the classic CSMA/CD algorithm for coaxial Ethernet~\cite{ieee802-3}.
Another approach is for clients to submit transactions
to a gossip network that includes a set of intermediating back-end proxies,
each of which collect many clients' transactions into blocks
to propose and commit in batches, as in Bitcoin~\cite{nakamoto08bitcoin}.
This way, it does not matter to clients which proxy's proposal
wins a given round
provided \emph{some} proxy includes the client's transaction in a block.

\subsection{Implementation optimizations}
\label{sec:od:opt}

\xxx{ this could easily be moved to the appendix too if needed. }

The above complexity analysis assumes that
a client is already ``caught up'' to the servers' recent state.
Implementations can enable 
a freshly-started or out-of-date client can catch up efficiently,
with effort and communication logarithmic rather than linear
in the history size,
by writing summaries or ``forward pointers'' to the key-value stores
enabling future clients to skip forward over exponentially-increasing distances,
as in skip lists~\cite{pugh90skip} or skipchains~\cite{nikitin17chainiac}.

\xxx{ practical: reduce round-trip latencies by write batching.
A batching protocol design well-suited to this use-case
is for clients to submit a sequence of key-value pairs at once,
each which the server attempts to write in-order,
stopping at the first write that fails
(typically because another client already wrote a value at that key).
The server then reads all the keys specified in the batch in-order,
again stopping at the first read that fails
(typically because no value has yet been written to that key),
and returns the actual values for all successfully-read keys.
}

\com{	Byzantine probably needs to wait for the next paper
\input{byz/model}

}

\section{Limitations and Future Work}
\label{sec:lim}

The \qsc and \tlc protocols developed here
have many limitations,
most notably tolerating only crash-stop node
failures~\cite{schlichting83fail,cachin11introduction}.
It appears readily feasible to extend \qsc and \tlc
to tolerate Byzantine node failures
along lines already proposed informally~\cite{ford19threshold}.
Further, it seems promising to generalize the principles of \qsc and \tlc
to support quorum systems~\cite{malkhi98byzantine,hirt00player,lokhava19fast,cachin19asymmetric},
whose threat models assume that not just a simple threshold of nodes,
but more complex subsets,
might fail or be compromised.
Full formal development and analysis of Byzantine versions of \qsc and \tlc,
however,
remains for future work.

Most efficient asynchronous Byzantine consensus protocols
rely on threshold secret sharing
schemes~\cite{shamir79share,stadler96publicly,schoenmakers99simple}
to provide shared
randomness~\cite{canetti93fast,cachin05random,syta17scalable}
and/or efficient threshold
signing~\cite{syta16keeping,abraham19vaba}.
Setting up these schemes asynchronously without a trusted dealer, however,
requires distributed key generation or
DKG~\cite{cachin02asynchronous,zhou05apss,kate09distributed,kokoris19bootstrapping}.
The \tlc framework appears applicable to efficient DKG
as well~\cite{ford19threshold},
but detailed development and analysis of this application of \tlc
is again left for future work.

While this paper focuses on implementing consensus
in a fashion functionally-equivalent to [Multi-]Paxos or Raft,
it remains to be determined how best to implement
closely-related primitives such as
atomic broadcast~\abcite{cachin11introduction}{cristian95atomic,defago04total,marandi10ring}
in the \tlc framework.
For example,
\qsc as formulated here guarantees only that each round has a reasonable chance
of committing {\em some} node's proposal in that round --
but does not guarantee that any {\em particular} node's proposals
have a ``fair'' chance, or even even are {\em ever},
included in the final total order.
Indeed, a node that is consistently much slower than the others
will never see its proposals chosen for commitment.
An atomic broadcast protocol, in contrast, should guarantee
that {\em all} messages submitted by {\em any} correct node
are {\em eventually} included in the final total order.
The ``fairness'' or ``eventual-inclusion'' guarantees
required for atomic broadcast
are also closely-related to properties like {\em chain quality}
recently explored in the context of blockchains~\abcite{pass17hybrid,pass17fruitchains}{bentov16snow}.

While the algorithms described above
and their fundamental complexity-theoretic characteristics
suggest that \qsc and \tlc should yield 
simple and efficient protocols in practice,
these properties remain to be confirmed empirically
with fully-functional prototypes and rigorous experimental evaluation.
In particular,
we would like to see systematic user studies of the
difficulty of implementing \qsc/\tlc in comparison with 
traditional alternatives,
similar to the studies that have been done
on Raft~\abcite{ongaro14search}{howard15raft}.
In addition, while we have decades of experience optimizing
implementations of Paxos for maximum performance and efficiency
in deployment environments,
it will take time and experimentation to determine
how these lessons do or don't translate, or must be adapted,
to apply to practical implementations of \qsc/\tlc.

\section{Related Work}
\label{sec:rel}

This section summarizes related work,
focusing first on \tlc in relation to classic logical clocks, and 
then on \qsc in relation to other consensus protocols, first asynchronous and then those
specifically designed with simplicity in mind. 

\paragraph{Logical clocks and virtual time}

\tlc is inspired by
classic notions of logical time,
such as Lamport clocks~\cite{lamport78time,raynal92about},
vector clocks~\abcite{fidge88timestamps}{fischer82sacrificing,liskov86highly,mattern89virtual,fidge91logical} 
and matrix clocks~\abcite{wuu84efficient,drummond03reducing}{sarin87discarding,ruget94cheaper,raynal92about}.
Prior work has used logical clocks and virtual time for purposes
such as discrete event simulation and rollback~\cite{jefferson85virtual},
verifying cache coherence protocols~\cite{plakal98lamport},
and temporal proofs for digital ledgers~\cite{hughes17radix}.
We are not aware of prior work defining a threshold logical clock abstraction
or using it to build asynchronous consensus, however.


Conceptually analogous to \tlc,
Awerbuch's synchronizers~\cite{awerbuch85complexity}
are intended to simplify the design of distributed algorithms
by presenting a synchronous abstraction atop an asynchronous network.
Awerbuch's synchronizers assume a fully-reliable system, however,
tolerating no failures in participating nodes.
\tlc's purpose might therefore be reasonably described as building
{\em fault-tolerant synchronizers}.

The basic threshold communication patterns \tlc employs
have appeared in numerous protocols in various forms,
such as classic reliable
broadcast~\abcite{bracha84asynchronous}{bracha85asynchronous,reiter94secure}.
Witnessed \tlc is inspired by
threshold signature schemes~\abcite{shoup00practical}{boldyreva03threshold},
signed echo broadcast~\cite{reiter94secure,cachin01secure,abraham19vaba},
and witness cosigning protocols~\cite{syta16keeping,nikitin17chainiac}.
We are not aware of prior work to develop or use a form of logical clock
based on these threshold primitives, however,
or to use them for purposes such as asynchronous consensus.

\com{	vague and wordy, not sure what this adds... -BF
Lastly, \tlc provides for a separation of handling of time from a process, \eg, consensus, 
built atop. Specifically, it allows upper layers to take advantage of synchrony assumptions, such as 
requiring lower corruptions thresholds or using simpler designs, while 
operating over an asynchronous network. 
\com{
In a similar quest to decompose assumptions from designs, 
other works~\cite{gafni98round,loss18combining,malkhi19flexible,blum19synchronous} 
circumvented the need to make an uniform assumption for all layers of the system 
by designing hybrid protocols that can retain {\em different} notions of safety under different assumptions.
}
}

\paragraph{Asynchronous consensus protocols}

The FLP theorem~\cite{fischer85impossibility}
implies that consensus protocols
must sacrifice one of safety, liveness, asynchrony, or determinism.
\com{
Paxos~\cite{lamport98parttime,lamport01paxos}
and its leader-based derivatives for
the fail-stop model~\cite{boichat03deconstructing,ongaro14search,renesse15paxos,howard15raft}
sacrifice asynchrony by relying on timeouts to ensure progress,
leaving them vulnerable to performance and DoS
attacks~\cite{clement09making,amir11byzantine}.
}
\qsc sacrifices determinism
and implements a probabilistic approach to consensus.
Randomness has been used in consensus protocols in various ways:
Some use private coins that nodes flip independently
but require time exponential in group
size~\abcite{bracha84asynchronous}{ben-or83another,moniz06randomized},
assume that the network embodies randomness
in the form of a {\em fair scheduler}~\cite{bracha85asynchronous},
or rely on shared coins~\abcite{cachin05random,miller16honey,abraham19vaba,syta17scalable}{rabin83randomized,ben-or85fast,canetti93fast,cachin01secure,cachin02secure,friedman05simple,correia06consensus,correia11byzantine,mostefaoui14signature,duan18beat}.
Shared coins require complex setup protocols, however,
a problem as hard as asynchronous consensus
itself~\cite{cachin02asynchronous,zhou05apss,kate09distributed,kokoris19bootstrapping}.
\com{	some of these are probably only applicable to Byzantine protocols...
which, however, often 
assume a trusted dealer~\cite{cachin02asynchronous} 
a partially-synchronous network~\cite{kate12distributed} 
a weakened fault tolerance threshold~\cite{canetti98fast} 
or weakened termination guarantees~\cite{bangalore18almost}. 
}
\qsc in contrast requires
only private randomness and private communication channels.

\qsc's consensus approach,
where each node maintains its own history
but adopts those of others so as to converge statistically,
is partly inspired by randomized blockchain consensus
protocols~\abcite{nakamoto08bitcoin,gilad17algorand}{kiayias16ouroboros,abraham18dfinity},
which rely on synchrony assumptions however.
\com{These blockchain protocols rely on synchrony assumptions, however,
such as the essential block interval parameter that paces
Bitcoin's proof-of-work~\cite{gervais16security}.}
\qsc in a sense provides Bitcoin-like consensus
using \tlc for fully-asynchronous pacing and
replacing Bitcoin's ``longest chain'' rule with a ``highest priority'' rule.
\com{\qsc builds on the classic techniques of
tamper-evident logging~\cite{schneier99secure,crosby09efficient},
timeline entanglement~\cite{maniatis02secure},
and accountable state
machines~\cite{haeberlen07peerreview,haeberlen10accountable}
for general protection against Byzantine node behavior.
Several recent DAG-based blockchain consensus
protocols~\cite{lewenberg15inclusive,baird16hashgraph,popov18tangle,danezis18blockmania}
reinvent specialized variants of these techniques.
}

\com{
\xxx{PJ: we should probably discuss~\cite{moniz06randomized}}
BF: worked it into the citations.  It's in the line of complex stack appraoches
like cachin05 and miller16 that uses binary consensus,
in this case relying on the exponential-time randomized approach
with private coins.
I don't see a need for discussion specifically of this paper. -BF
}

\paragraph{Consensus protocols designed for simplicity}

Consensus protocols, such as the classic
(Multi-)Paxos~\cite{lamport98parttime}, are notoriously difficult
to understand, implement, and reason about.
This holds expecially for those
variants that run atop asynchronous networks, can handle Byzantine
faults, or try to tackle
both~\cite{castro99practical,cachin02secure,cachin05random,moniz06randomized,miller16honey}.

(Multi-)Paxos, despite being commonly taught and used in real-world deployments, required a number of additional 
attempts to clarify its design and 
further modifications to adapt it for practical
applications~\abcite{lamport01paxos,lampson01abcd,van15paxos}{mazieres07paxos,chandra07paxos,kirsch08paxos,cui15paxos}.
The intermingling of agreement and network synchronization appears to be a source of
algorithmic complexity that has not been addressed adequately
in past generations of consensus protocols,
resulting in complex leader-election and view-change (sub-)protocols
and restrictions to partial
synchrony~\abcite{ongaro14search}{howard15raft}
\com{\qsc's design is most closely related to Multi-Paxos~\cite{chand16formal}, where nodes 
agree on a series of values by using a single run of
Paxos~\cite{lamport98parttime} for each value.
Multi-Paxos, however, requires a leader to improve performance, allowing to skip {\em prepare}
phases so long the leader is alive.
(Multi-)Paxos, despite being commonly taught and used in real-world deployments, 
is difficult to understand, required a number of additional 
attempts to clarify its
design~\cite{lampson96build,lamport01paxos,lampson01abcd,garcia18paxos}, and 
further modifications to adapt it for practical
applications~\cite{mazieres07paxos,chandra07paxos,kirsch08paxos,cui15paxos,van15paxos}.
Lastly, there is no single, agreed-upon specification of Multi-Paxos and many deployments use custom
designs and implementations~\cite{burrows06chubby,chandra07paxos,hunt10zookeeper,junqueira11zab}.}

In its aim for simplicity and understandability, \qsc is closely related to
Raft~\cite{ongaro14search},
which however assumes a partially-synchronous network and relies on a leader.
\com{Both Multi-Paxos and Raft assume a partially synchronous network.}
\qsc appears to be the first practical
yet conceptually simple asynchronous consensus protocol
that depends on neither leaders nor common coins,
making it more robust to slow leaders or network denial-of-service attacks. 
The presented approach is relatively clean and simple in part
due to the decomposition of the agreement problem (via \qsc) 
from that of network asynchrony (via \tlc).


\com{
\xxx{ES: needs  to be added somewhere... "Ouroboros Chronos: Permissionless Clock Synchronization via Proof-of-Stake”
\url{https://eprint.iacr.org/2019/838.pdf}}

Ouroboros~\cite{kiayias16ouroboros},
Ouroboros Chronos~\cite{badertscher19ouroboros}
}

\com{
This ``propose, gossip, decide'' approach relates to
recent DAG-based blockchain consensus
proposals~\cite{lewenberg15inclusive,baird16hashgraph,popov18tangle,danezis18blockmania},
which reinvent and apply classic principles
of secure timeline entanglement~\cite{maniatis02secure}
and accountable state machines~\cite{haeberlen07peerreview,haeberlen10accountable}.
The approach to consensus we propose attempts to clarify and systematize
this direction in light of existing tools and abstractions.
}

\section{Conclusion}
\label{sec:concl}

This paper has presented \qsc, the first asynchronous consensus protocol
arguably simpler than current partially-synchronous workhorses
like Paxos and Raft.
\qsc requires neither leader election, view changes, nor common coins,
and cleanly decomposes the consensus problem itself
from that of handling network asynchrony.
With appropriate implementation optimizations,
\qsc completes in $O(1)$ expected rounds per agreement,
incurring $O(n^3)$ communication bits in a broadcast-based group,
or $O(n^2)$ bits per client-driven transaction
in an on-demand implementation approach.

\subsection*{Acknowledgments}

This work benefitted
from discussion with numerous colleagues:
in particular 
Eleftherios Kokoris-Kogias,
Enis Ceyhun Alp,
Manuel Jos\'{e} Ribeiro Vidigueira,
Nicolas Gailly,
Cristina Basescu,
Timo Hanke,
Mahnush Movahedi,
and Dominic Williams.

This ongoing research was facilitated in part by financial support from
DFINITY, AXA, Handshake, and EPFL.
DFINITY's support in paticular,
which funded a joint project to analyze, improve, and formalize
its consensus protocol,
provided a key early impetus to explore randomized consensus protocols further.


\bibliographystyle{plain}
\arxiv{
\bibliography{os,net,sec,theory,lang}
}{
\bibliography{main}

\begin{thebibliography}{100}

\bibitem{abraham17revisiting}
Ittai Abraham, Guy Gueta, Dahlia Malkhi, Lorenzo Alvisi, Rama Kotla, and
  Jean-Philippe Martin.
\newblock \href{https://arxiv.org/pdf/1712.01367.pdf}{Revisiting Fast Practical
  Byzantine Fault Tolerance}, December 2017.

\bibitem{abraham18dfinity}
Ittai Abraham, Dahlia Malkhi, Kartik Nayak, and Ling Ren.
\newblock \href{https://eprint.iacr.org/2018/1153}{Dfinity Consensus,
  Explored}.
\newblock Cryptology ePrint Archive, Report 2018/1153, November 2018.

\bibitem{abraham19vaba}
Ittai Abraham, Dahlia Malkhi, and Alexander Spiegelman.
\newblock
  \href{https://dahliamalkhi.files.wordpress.com/2019/06/vaba-podc2019-2.pdf}{Asymptotically
  Optimal Validated Asynchronous Byzantine Agreement}.
\newblock In {\em \bibconf{PODC}{ACM Symposium on Principles of Distributed
  Computing}}, July 2019.

\bibitem{amir11byzantine}
Yair Amir, Brian Coan, Jonathan Kirsch, and John Lane.
\newblock {Prime}: {Byzantine} replication under attack.
\newblock {\em IEEE Transactions on Dependable and Secure Computing},
  8(4):564--577, July 2011.

\bibitem{armstrong13programming}
Joe Armstrong.
\newblock {\em Programming Erlang: Software for a Concurrent World}.
\newblock Pragmatic Bookshelf, 2 edition, October 2013.

\bibitem{aspnes03randomized}
James Aspnes.
\newblock
  \href{https://link.springer.com/article/10.1007%2Fs00446-002-0081-5}{Randomized
  protocols for asynchronous consensus}.
\newblock {\em Distributed Computing}, 16(2--3):165--175, September 2003.

\bibitem{awerbuch85complexity}
Baruch Awerbuch.
\newblock \href{https://dl.acm.org/citation.cfm?id=4227}{Complexity of Network
  Synchronization}.
\newblock {\em Journal of the Association for Computing Machinery},
  32(4):804--823, October 1985.

\bibitem{ben-or83another}
Michael Ben-Or.
\newblock Another advantage of free choice: Completely asynchronous agreement
  protocols.
\newblock In {\em Principles of Distributed Computing (PODC)}, 1983.

\bibitem{ben-or85fast}
Michael Ben-Or.
\newblock \href{https://dl.acm.org/citation.cfm?id=323609}{Fast Asynchronous
  Byzantine Agreement (Extended Abstract)}.
\newblock In {\em \bibconf[4th]{PODC}{Principles of Distributed Computing}},
  pages 149--151, 1985.

\bibitem{bentov16snow}
Iddo Bentov, Rafael Pass, and Elaine Shi.
\newblock \href{http://eprint.iacr.org/2016/919}{Snow White: Provably Secure
  Proofs of Stake}.
\newblock Cryptology ePrint Archive, Report 2016/919, 2016.

\bibitem{boichat03deconstructing}
Romain Boichat, Partha Dutta, Svend Frølund, and Rachid Guerraoui.
\newblock Deconstructing {Paxos}.
\newblock {\em ACM SIGACT News}, 34(1), March 2003.

\bibitem{boldyreva03threshold}
Alexandra Boldyreva.
\newblock
  \href{https://link.springer.com/chapter/10.1007/3-540-36288-6_3}{Threshold
  Signatures, Multisignatures and Blind Signatures Based on the
  {Gap-Diffie-Hellman-Group} Signature Scheme}.
\newblock In {\em \bibconf[6th]{PKC}{International Workshop on Practice and
  Theory in Public Key Cryptography}}, January 2003.

\bibitem{bracha84asynchronous}
Gabriel Bracha.
\newblock \href{https://dl.acm.org/citation.cfm?id=806743}{An asynchronous
  [(n-1)/3]-Resilient Consensus Protocol}.
\newblock In {\em \bibconf[3rd]{PODC}{ACM Symposium on Principles of
  Distributed Computing}}, 1984.

\bibitem{bracha85asynchronous}
Gabriel Bracha and Sam Toueg.
\newblock \href{https://dl.acm.org/citation.cfm?id=214134}{Asynchronous
  Consensus and Broadcast Protocols}.
\newblock {\em Journal of the Association for Computing Machinery (JACM)},
  32(4), 1985.

\bibitem{cachin02asynchronous}
Christian Cachin, Klaus Kursawe, Anna Lysyanskaya, and Reto Strobl.
\newblock \href{https://dl.acm.org/citation.cfm?id=586124}{Asynchronous
  Verifiable Secret Sharing and Proactive Cryptosystems}.
\newblock In {\em \bibconf[9th]{CCS}{ACM Conference on Computer and
  Communications Security}}, November 2002.

\bibitem{cachin01secure}
Christian Cachin, Klaus Kursawe, Frank Petzold, and Victor Shoup.
\newblock \href{http://www.zurich.ibm.com/~cca/papers/abc.pdf}{Secure and
  Efficient Asynchronous Broadcast Protocols}.
\newblock In {\em Advances in Cryptology (CRYPTO)}, 2001.

\bibitem{cachin05random}
Christian Cachin, Klaus Kursawe, and Victor Shoup.
\newblock Random oracles in {Constantinople}: Practical asynchronous byzantine
  agreement using cryptography.
\newblock {\em Journal of Cryptology}, 18(3):219--246, 2005.

\bibitem{cachin02secure}
Christian Cachin and Jonathan~A. Poritz.
\newblock \href{https://cachin.com/cc/papers/sintra.pdf}{Secure
  Intrusion-tolerant Replication on the Internet}.
\newblock In {\em \bibconf{DSN}{Dependable Systems and Networks}}, June 2002.

\bibitem{cachin11introduction}
Christian Cachin and Rachid Guerraoui~Lu\'is Rodrigues.
\newblock {\em Introduction to Reliable and Secure Distributed Programming}.
\newblock Springer, February 2011.

\bibitem{cachin19asymmetric}
Christian Cachin and Bj\"orn Tackmann.
\newblock \href{https://arxiv.org/pdf/1906.09314.pdf}{Asymmetric Distributed
  Trust}, June 2019.

\bibitem{canetti93fast}
Ran Canetti and Tal Rabin.
\newblock \href{https://dl.acm.org/citation.cfm?id=167105}{Fast Asynchronous
  Byzantine Agreement with Optimal Resilience}.
\newblock In {\em \bibconf[25th]{STOC}{ACM Symposium on Theory of computing}},
  pages 42--51, May 1993.

\bibitem{castro99practical}
Miguel Castro and Barbara Liskov.
\newblock
  \href{http://css.csail.mit.edu/6.824/2014/papers/castro-practicalbft.pdf}{Practical
  {Byzantine} Fault Tolerance}.
\newblock In {\em \bibconf[3rd]{OSDI}{USENIX Symposium on Operating Systems
  Design and Implementation}}, February 1999.

\bibitem{chand16formal}
Saksham Chand, Yanhong~A. Liu, and Scott~D. Stoller.
\newblock
  \href{https://link.springer.com/chapter/10.1007/978-3-319-48989-6_8}{Formal
  Verification of Multi-Paxos for Distributed Consensus}.
\newblock In {\em \bibconf[21st]{FM}{International Symposium on Formal
  Methods}}, November 2016.

\bibitem{chand19formal}
Saksham Chand, Yanhong~A. Liu, and Scott~D. Stoller.
\newblock \href{https://arxiv.org/abs/1606.01387}{Formal Verification of
  Multi-Paxos for Distributed Consensus}, November 2019.

\bibitem{chandra07paxos}
Tushar~D Chandra, Robert Griesemer, and Joshua Redstone.
\newblock Paxos made live: an engineering perspective.
\newblock In {\em Proceedings of the twenty-sixth annual ACM symposium on
  Principles of distributed computing}, pages 398--407, 2007.

\bibitem{clement09making}
Allen Clement, Edmund~L Wong, Lorenzo Alvisi, Michael Dahlin, and Mirco
  Marchetti.
\newblock
  \href{https://www.usenix.org/legacy/events/nsdi09/tech/full_papers/clement/clement.pdf}{Making
  {Byzantine} Fault Tolerant Systems Tolerate Byzantine Faults}.
\newblock In {\em 6th USENIX Symposium on Networked Systems Design and
  Implementation}, 2009.

\bibitem{correia06consensus}
Miguel Correia, Nuno~Ferreira Neves, and Paulo Veríssimo.
\newblock From consensus to atomic broadcast: Time-free {Byzantine}-resistant
  protocols without signatures.
\newblock {\em The Computer Journal}, 49(1), January 2006.

\bibitem{correia11byzantine}
Miguel Correia, Giuliana~Santos Veronese, Nuno~Ferreira Neves, and Paulo
  Verissimo.
\newblock {Byzantine} consensus in asynchronous message-passing systems: a
  survey.
\newblock {\em International Journal of Critical Computer-Based Systems},
  2(2):141--161, July 2011.

\bibitem{cristian95atomic}
Flaviu Cristian, Houtan Aghili, Ray Strong, and Danny Dolev.
\newblock
  \href{https://www.sciencedirect.com/science/article/pii/S0890540185710607}{Atomic
  Broadcast: From Simple Message Diffusion to Byzantine Agreement}.
\newblock {\em Information and Computation}, 118(1):158--179, April 1995.

\bibitem{crosby09efficient}
Scott~A. Crosby and Dan~S. Wallach.
\newblock Efficient data structures for tamper-evident logging.
\newblock In {\em USENIX Security Symposium}, August 2009.

\bibitem{cui15paxos}
Heming Cui, Rui Gu, Cheng Liu, Tianyu Chen, and Junfeng Yang.
\newblock Paxos made transparent.
\newblock In {\em Proceedings of the 25th Symposium on Operating Systems
  Principles}, pages 105--120, 2015.

\bibitem{cloc}
Al~Danial.
\newblock {Counting Lines of Code}.
\newblock \url{http://cloc.sourceforge.net/}.

\bibitem{defago04total}
Xavier D\'efago, Andr\'e Schiper, and P\'eter Urb\'an.
\newblock \href{https://dl.acm.org/doi/abs/10.1145/1041680.1041682}{Total Order
  Broadcast and Multicast Algorithms:Taxonomy and Survey}.
\newblock {\em ACM Computing Surveys}, December 2004.

\bibitem{drummond03reducing}
L\'{u}cia M.~A. Drummond and Valmir~C. Barbosa.
\newblock \href{https://arxiv.org/pdf/cs/0309042.pdf}{On reducing the
  complexity of matrix clocks}.
\newblock {\em Parallel Computing}, 29(7):895--905, July 2003.

\bibitem{duan18beat}
Sisi Duan, Michael~K. Reiter, and Haibin Zhang.
\newblock \href{https://dl.acm.org/citation.cfm?id=3243812}{BEAT: Asynchronous
  BFT Made Practical}.
\newblock In {\em \bibconf{CCS}{Computer and Communications Security}}, October
  2018.

\bibitem{fidge91logical}
Colin Fidge.
\newblock Logical time in distributed computing systems.
\newblock {\em IEEE Computer}, 24(8):28--33, August 1991.

\bibitem{fidge88timestamps}
Colin~J. Fidge.
\newblock
  \href{https://zoo.cs.yale.edu/classes/cs426/2012/lab/bib/fidge88timestamps.pdf}{Timestamps
  in Message-Passing Systems That Preserve the Partial Ordering}.
\newblock In {\em 11th \bibbrev{ACSC}{Australian Computer Science Conference}},
  1988.

\bibitem{fischer85impossibility}
Michael~J Fischer, Nancy~A Lynch, and Michael~S Paterson.
\newblock
  \href{https://groups.csail.mit.edu/tds/papers/Lynch/jacm85.pdf}{Impossibility
  of distributed consensus with one faulty process}.
\newblock {\em Journal of the ACM (JACM)}, 32(2):374--382, 1985.

\bibitem{fischer82sacrificing}
Michael~J. Fischer and Alan Michael.
\newblock Sacrificing serializability to attain high availability of data in an
  unreliable network.
\newblock In {\em Symposium on Principles of Database Systems (PODS)}, pages
  70--75, March 1982.

\bibitem{ford07structured}
Bryan Ford.
\newblock Structured streams: a new transport abstraction.
\newblock In {\em \bibbrev{SIGCOMM}{ACM SIGCOMM}}, August 2007.

\bibitem{ford19threshold}
Bryan Ford.
\newblock \href{https://arxiv.org/abs/1907.07010}{Threshold Logical Clocks for
  Asynchronous Distributed Coordination and Consensus}, July 2019.

\bibitem{friedman05simple}
Roy Friedman, Achour Mostefaoui, and Michel Raynal.
\newblock Simple and efficient oracle-based consensus protocols for
  asynchronous {Byzantine} systems.
\newblock {\em IEEE Transactions on Dependable and Secure Computing}, 2(1),
  January 2005.

\bibitem{gilad17algorand}
Yossi Gilad, Rotem Hemo, Silvio Micali, Georgios Vlachos, and Nickolai
  Zeldovich.
\newblock \href{https://dl.acm.org/authorize?N47148}{Algorand: Scaling
  Byzantine Agreements for Cryptocurrencies}, October 2017.

\bibitem{rfc7530}
T.~Haynes and D.~Noveck.
\newblock \href{https://tools.ietf.org/html/rfc7530}{Network File System
  {(NFS)} version 4 Protocol}, March 2015.
\newblock RFC 7530.

\bibitem{hirt00player}
Martin Hirt and Ueli Maurer.
\newblock \href{https://link.springer.com/article/10.1007/s001459910003}{Player
  Simulation and General Adversary Structures in Perfect Multiparty
  Computation}.
\newblock 31:31--60, April 2000.

\bibitem{rfc1036}
M.~Horton and R.~Adams.
\newblock \href{https://tools.ietf.org/html/rfc1036}{Standard for Interchange
  of {USENET} Messages}, December 1987.
\newblock RFC 1036.

\bibitem{howard15raft}
Heidi Howard, Malte Schwarzkopf, Anil Madhavapeddy, and Jon Crowcroft.
\newblock {Raft} refloated: Do we have consensus?
\newblock {\em ACM SIGOPS Operating Systems Review}, 49(1):12--21, January
  2015.

\bibitem{hughes17radix}
Dan Hughes.
\newblock \href{https://papers.radixdlt.com/tempo/}{Radix -- Tempo}, September
  2017.

\bibitem{ieee802-3}
IEEE.
\newblock Std 802.3-2005, December 2005.

\bibitem{jefferson85virtual}
David~R. Jefferson.
\newblock Virtual time.
\newblock {\em ACM Transactions on Programming Languages and Systems}, 7(3),
  July 1985.

\bibitem{kafka}
kafka - a distributed streaming platform.
\newblock \url{http://kafka.apache.org.}

\bibitem{kate09distributed}
Aniket Kate and Ian Goldberg.
\newblock
  \href{https://ieeexplore.ieee.org/abstract/document/5158416/}{Distributed Key
  Generation for the Internet}.
\newblock In {\em \bibconf[29th]{ICDCS}{International Conference on Distributed
  Computing Systems}}, pages 119--128. IEEE, June 2009.

\bibitem{kiayias16ouroboros}
Aggelos Kiayias, Alexander Russell, Bernardo David, and Roman Oliynykov.
\newblock \href{http://eprint.iacr.org/2016/889}{Ouroboros: A Provably Secure
  Proof-of-Stake Blockchain Protocol}.
\newblock Cryptology ePrint Archive, Report 2016/889, 2016.

\bibitem{kirsch08paxos}
Jonathan Kirsch and Yair Amir.
\newblock Paxos for system builders: An overview.
\newblock In {\em Proceedings of the 2nd Workshop on Large-Scale Distributed
  Systems and Middleware}, pages 1--6, 2008.

\bibitem{kokoris16enhancing}
Eleftherios Kokoris-Kogias, Philipp Jovanovic, Nicolas Gailly, Ismail Khoffi,
  Linus Gasser, and Bryan Ford.
\newblock \href{http://arxiv.org/abs/1602.06997}{Enhancing Bitcoin Security and
  Performance with Strong Consistency via Collective Signing}.
\newblock In {\em Proceedings of the 25th USENIX Conference on Security
  Symposium}, 2016.

\bibitem{kokoris19bootstrapping}
Eleftherios Kokoris-Kogias, Alexander Spiegelman, Dahlia Malkhi, and Ittai
  Abraham.
\newblock \href{https://eprint.iacr.org/2019/1015.pdf}{Bootstrapping Consensus
  Without Trusted Setup: Fully Asynchronous Distributed Key Generation}.
\newblock Cryptology ePrint Archive Report 2019/1015, September 2019.

\bibitem{lamport78time}
Leslie Lamport.
\newblock Time, clocks, and the ordering of events in a distributed system.
\newblock {\em Communications of the ACM}, 21(7):558--565, July 1978.

\bibitem{lamport98parttime}
Leslie Lamport.
\newblock \href{https://dl.acm.org/citation.cfm?id=279229}{The Part-Time
  Parliament}.
\newblock {\em ACM Transactions on Computer Systems}, 16(2):133--169, May 1989.

\bibitem{lamport01paxos}
Leslie Lamport.
\newblock Paxos made simple.
\newblock {\em ACM SIGACT News}, 32(4):51--58, December 2001.

\bibitem{lampson01abcd}
Butler~W. Lampson.
\newblock
  \href{https://citeseerx.ist.psu.edu/viewdoc/download?doi=10.1.1.595.4829&rep=rep1&type=pdf}{The
  ABCD’s of Paxos}.
\newblock In {\em Proceedings of the 20th Annual ACM Symposium on Principles of
  Distributed Computing}, PODC ’01, 2001.

\bibitem{leesatapornwongsa16taxdc}
Tanakorn Leesatapornwongsa, Jeffrey~F. Lukman, Shan Lu, and Haryadi~S. Gunawi.
\newblock \href{https://ucare.cs.uchicago.edu/pdf/asplos16-TaxDC.pdf}{TaxDC: A
  Taxonomy of Non-Deterministic Concurrency Bugs in Datacenter Distributed
  Systems}.
\newblock In {\em \bibconf[21th]{ASPLOS}{International Conference on
  Architectural Support for Programming Languages and Operating Systems}},
  April 2016.

\bibitem{liskov86highly}
Barbara Liskov and Rivka Ladin.
\newblock Highly-available distributed services and fault-tolerant distributed
  garbage collection.
\newblock In {\em Principles of Distributed Computing}, pages 29--39, August
  1986.

\bibitem{lokhava19fast}
Marta Lokhava, Giuliano Losa, David Mazi\`eres, Graydon Hoare, Nicolas Barry,
  Eli Gafni, Jonathan Jove, Rafa\l{} Malinowsky, and Jed McCaleb.
\newblock \href{https://dl.acm.org/doi/abs/10.1145/3341301.3359636}{Fast and
  secure global payments with Stellar}.
\newblock In {\em \bibconf[27th]{SOSP}{ACM Symposium on Operating System
  Principles}}, pages 80--96, October 2019.

\bibitem{luckie15resilience}
Matthew Luckie, Robert Beverly, Tiange Wu, Mark Allman, and kc~claffy.
\newblock \href{https://dl.acm.org/doi/abs/10.1145/2815675.2815700}{Resilience
  of Deployed TCP to Blind Attacks}.
\newblock In {\em \bibbrev{IMC}{Internet Measurement Conference}}, 2015.

\bibitem{malkhi98byzantine}
Dahlia Malkhi and Michael Reiter.
\newblock \href{https://dl.acm.org/doi/abs/10.1007/s004460050050}{Byzantine
  Quorum Systems}.
\newblock {\em Distributed Computing}, 11(4), October 1998.

\bibitem{marandi10ring}
Parisa~Jalili Marandi, Marco Primi, Nicolas Schiper, and Fernando Pedone.
\newblock
  \href{https://www.inf.usi.ch/faculty/pedone/Paper/2010/2010DSN.pdf}{Ring
  Paxos: A High-Throughput Atomic Broadcast Protocol}.
\newblock In {\em \bibconf{DSN}{Dependable Systems and Networks}}, June 2010.

\bibitem{mattern89virtual}
Friedemann Mattern.
\newblock
  \href{http://citeseerx.ist.psu.edu/viewdoc/download?doi=10.1.1.63.4399&rep=rep1&type=pdf}{Virtual
  Time and Global States of Distributed Systems}.
\newblock In {\em International Workshop on Parallel and Distributed
  Algorithms}, 1989.

\bibitem{mazieres07paxos}
David Mazieres.
\newblock Paxos made practical, 2007.

\bibitem{miller16honey}
Andrew Miller, Yu~Xia, Kyle Croman, Elaine Shi, and Dawn Song.
\newblock \href{https://eprint.iacr.org/2016/199.pdf}{The Honey Badger of BFT
  Protocols}.
\newblock In {\em \bibconf{CCS}{Computer and Communications Security}}, 2016.

\bibitem{moniz06randomized}
Henrique Moniz, Nuno~Ferreira Neves, Miguel Correia, and Paulo Verissimo.
\newblock
  \href{https://www.gsd.inesc-id.pt/~mpc/pubs/hmoniz-randomized.pdf}{Randomized
  Intrusion-Tolerant Asynchronous Services}.
\newblock In {\em Proceedings of the International Conference on Dependable
  Systems and Networks}, DSN ’06, pages 568--577, 2006.

\bibitem{mostefaoui14signature}
Achour Most\'efaoui, Hamouma Moumen, and Michel Raynal.
\newblock \href{https://dl.acm.org/citation.cfm?id=2611468}{Signature-Free
  Asynchronous Byzantine Consensus with $t < n/3$ and $O(n^2)$ Messages}.
\newblock In {\em \bibconf{PODC}{Principles of Distributed Computing}}, July
  2014.

\bibitem{nakamoto08bitcoin}
Satoshi Nakamoto.
\newblock \href{https://bitcoin.org/bitcoin.pdf}{Bitcoin: A Peer-to-Peer
  Electronic Cash System}, 2008.

\bibitem{nats}
Nats - the cloud native messaging system.
\newblock \url{http://nats.io}.

\bibitem{nikitin17chainiac}
Kirill Nikitin, Eleftherios Kokoris-Kogias, Philipp Jovanovic, Nicolas Gailly,
  Linus Gasser, Ismail Khoffi, Justin Cappos, and Bryan Ford.
\newblock
  \href{https://www.usenix.org/conference/usenixsecurity17/technical-sessions/presentation/nikitin}{{CHAINIAC}:
  Proactive Software-Update Transparency via Collectively Signed Skipchains and
  Verified Builds}.
\newblock In {\em 26th {USENIX} Security Symposium}, pages 1271--1287, 2017.

\bibitem{ongaro14search}
Diego Ongaro and John Ousterhout.
\newblock In search of an understandable consensus algorithm.
\newblock In {\em \bibconf{USENIX ATC}{USENIX Annual Technical Conference}},
  2014.

\bibitem{pass17hybrid}
Rafael Pass and Elaine Shi.
\newblock \href{http://drops.dagstuhl.de/opus/volltexte/2017/8004/}{Hybrid
  Consensus: Efficient Consensus in the Permissionless Model}.
\newblock In {\em \bibconf[31st]{DISC}{International Symposium on Distributed
  Computing}}, October 2017.

\bibitem{pass17fruitchains}
Rafael Pass and Elaine Shi.
\newblock \href{https://dl.acm.org/citation.cfm?id=3087809}{FruitChains: A Fair
  Blockchain}.
\newblock In {\em \bibconf{PODC}{ACM Symposium on Principles of Distributed
  Computing}}, July 2017.

\bibitem{plakal98lamport}
Manoj Plakal, Daniel~J. Sorin, Anne~E. Condon, and Mark~D. Hill.
\newblock \href{https://dl.acm.org/citation.cfm?doid=277651.277672}{Lamport
  Clocks: Verifying a Directory Cache-Coherence Protocol}.
\newblock In {\em \bibconf[10th]{SPAA}{Symposium on Parallel Algorithms and
  Architectures}}, June 1998.

\bibitem{pugh90skip}
William Pugh.
\newblock Skip lists: a probabilistic alternative to balanced trees.
\newblock {\em Communications of the ACM}, 33(6):668--676, 1990.

\bibitem{rabbitmq}
Rabbitmq - an open source multi-protocol messaging broker.
\newblock \url{https://www.rabbitmq.com/}.

\bibitem{rabin83randomized}
Michael~O. Rabin.
\newblock Randomized {Byzantine} generals.
\newblock In {\em Symposium on Foundations of Computer Science (SFCS)},
  November 1983.

\bibitem{raynal92about}
Michel Raynal.
\newblock About logical clocks for distributed systems.
\newblock {\em ACM SIGOPS Operating Systems Review}, 26(1), January 1992.

\bibitem{redis}
The redis pub sub.
\newblock \url{http://redis.io/topics/pubsub}.

\bibitem{reiter94secure}
Michael~K. Reiter.
\newblock \href{https://dl.acm.org/citation.cfm?id=191194}{Secure Agreement
  Protocols: Reliable and Atomic Group Multicast in Rampart}.
\newblock In {\em \bibconf[2nd]{CCS}{Computer and Communications Security}},
  1994.

\bibitem{renesse15paxos}
Robbert~Van Renesse and Deniz Altinbuken.
\newblock {Paxos} made moderately complex.
\newblock {\em ACM Computing Surveys (CSUR)}, 47(3), April 2015.

\bibitem{rfc8446}
E.~Rescorla.
\newblock \href{https://tools.ietf.org/html/rfc8446}{The Transport Layer
  Security ({TLS}) Protocol Version 1.3}.
\newblock RFC 8446, RFC Editor, August 2018.

\bibitem{ruget94cheaper}
Fr\'{e}d\'{e}ric Ruget.
\newblock
  \href{https://pdfs.semanticscholar.org/e060/433cd977fc11c29919bb01bc7c8b9500a144.pdf}{Cheaper
  matrix clocks}.
\newblock In {\em International Workshop on Distributed Algorithms (WDAG)},
  pages 355--369, September 1994.

\bibitem{sarin87discarding}
Sunil~K. Sarin and Nancy~A. Lynch.
\newblock Discarding obsolete information in a replicated database system.
\newblock {\em IEEE Transactions on Software Engineering}, SE-13(1), January
  1987.

\bibitem{schlichting83fail}
Richard~D Schlichting and Fred~B Schneider.
\newblock Fail-stop processors: an approach to designing fault-tolerant
  computing systems.
\newblock {\em ACM Transactions on Computer Systems (TOCS)}, 1(3):222--238,
  1983.

\bibitem{schneier99secure}
Bruce Schneier and John Kelsey.
\newblock Secure audit logs to support computer forensics.
\newblock {\em ACM Transactions on Information and System Security}, 2(2),
  1999.

\bibitem{schoenmakers99simple}
Berry Schoenmakers.
\newblock \href{https://link.springer.com/chapter/10.1007/3-540-48405-1_10}{A
  Simple Publicly Verifiable Secret Sharing Scheme and Its Application to
  Electronic Voting}.
\newblock In {\em \bibconf{CRYPTO}{IACR International Cryptology Conference}},
  pages 784--784, August 1999.

\bibitem{shamir79share}
Adi Shamir.
\newblock \href{https://cs.jhu.edu/~sdoshi/crypto/papers/shamirturing.pdf}{How
  to Share a Secret}.
\newblock {\em Communications of the ACM}, 22(11):612--613, 1979.

\bibitem{shoup00practical}
Victor Shoup.
\newblock
  \href{https://link.springer.com/content/pdf/10.1007/3-540-45539-6_15.pdf}{Practical
  Threshold Signatures}.
\newblock In {\em Eurocrypt}, May 2000.

\bibitem{stadler96publicly}
Markus Stadler.
\newblock
  \href{https://link.springer.com/content/pdf/10.1007/3-540-68339-9_17.pdf}{Publicly
  Verifiable Secret Sharing}.
\newblock In {\em Eurocrypt}, May 1996.

\bibitem{rfc4960}
R.~{Stewart, ed.}
\newblock Stream control transmission protocol, September 2007.
\newblock RFC 4960.

\bibitem{syta17scalable}
Ewa Syta, Philipp Jovanovic, Eleftherios Kokoris-Kogias, Nicolas Gailly, Linus
  Gasser, Ismail Khoffi, Michael~J. Fischer, and Bryan Ford.
\newblock
  \href{https://www.ieee-security.org/TC/SP2017/papers/413.pdf}{Scalable
  Bias-Resistant Distributed Randomness}.
\newblock In {\em 38th IEEE Symposium on Security and Privacy}, May 2017.

\bibitem{syta16keeping}
Ewa Syta, Iulia Tamas, Dylan Visher, David~Isaac Wolinsky, Philipp Jovanovic,
  Linus Gasser, Nicolas Gailly, Ismail Khoffi, and Bryan Ford.
\newblock \href{http://dedis.cs.yale.edu/dissent/papers/witness-abs}{Keeping
  Authorities ``Honest or Bust'' with Decentralized Witness Cosigning}.
\newblock In {\em 37th IEEE Symposium on Security and Privacy}, May 2016.

\bibitem{rfc793}
Transmission control protocol, September 1981.
\newblock RFC 793.

\bibitem{van15paxos}
Robbert Van~Renesse and Deniz Altinbuken.
\newblock Paxos made moderately complex.
\newblock {\em ACM Computing Surveys (CSUR)}, 47(3):1--36, 2015.

\bibitem{rfc908}
David Velten, Robert Hinden, and Jack Sax.
\newblock Reliable data protocol, July 1984.
\newblock RFC 908.

\bibitem{vinoski12concurrency}
Steve Vinoski.
\newblock \href{https://ieeexplore.ieee.org/document/6216341}{Concurrency and
  Message Passing in Erlang}.
\newblock {\em Computing in Science \& Engineering}, 14(6):24--34, November
  2012.

\bibitem{wuu84efficient}
Gene~T.J. Wuu and Arthur~J. Bernstein.
\newblock Efficient solutions to the replicated log and dictionary problems.
\newblock In {\em Principles of Distributed Computing}, pages 232--242, August
  1984.

\bibitem{zhou05apss}
Lidong Zhou, Fred~B. Schneider, and Robbert~Van Renesse.
\newblock \href{https://dl.acm.org/citation.cfm?id=1085127}{APSS: Proactive
  Secret Sharing in Asynchronous Systems}.
\newblock {\em ACM Transactions on Information and System Security (TISSEC)},
  8(3):259--286, August 2005.

\end{thebibliography}
}

\section*{Appendix}
\appendix
\section{Correctness Proofs}
\label{sec:proofs}

This appendix contains the proofs for the theorems in the main paper.

\subsection{Que Sera Consensus (\qsc)}
\label{sec:proofs:qsc}

This section contains correctness proofs
for the \qsc consensus algorithm (Section~\ref{sec:qsc}).

\lempreservation*


\begin{proof}
Let $h_{s_i}$ be the initial history of the round
starting at step $s$ on node $i$,
let $h'_{s_i}$ and $h''_{s_i}$ be $i$'s proposed and intermediate histories
in that round, respectively, and
let $(B'_{s_i},R'_{s_i})$ and $(B''_{s_i},R''_{s_i})$ be the sets
returned by the round's two \broadcast calls.
By \qsc's requirement that $t_b > 0$
and the \tsb's 
broadcast threshold property (Section~\ref{sec:tsb}),
the sets $B'_{s_i}$ and $B''_{s_i}$
returned by the round's two \broadcast calls are each nonempty.
By \tsb's 
receive threshold property,
the returned sets $R'_{s_i}$ and $R''_{s_i}$ are nonempty as well.
By message propagation through these \broadcast calls,
these sets consist solely of histories $h'_{sj}$
each proposed in the same round by some node $j$,
and each of which builds on $j$'s initial history $h_{s_j}$.
By induction over consensus rounds, therefore,
at step $s$ each node $i$'s history $h_{s_i}$
builds on some node $j$'s history $h_{s'j}$ at each earlier step $s' < s$.
That is, $h_{s'j}$ is a strict prefix of $h_{s_i}$.
\end{proof}

\lemagreement*
\begin{proof}
Agreement can be violated only if some node $j$
arrives at a different resulting history $h_{(s+2)_j} \ne h_{(s+2)_i}$.

Because $h_{(s+2)_i} \in B''_{s_i}$
and $B''_{s_i} \subseteq R''_{s_j}$ by \tsb's 
broadcast spread property,
$i$'s delivered history $h_{(s+2)_i}$ is also among the set of histories
from which $j$ chooses its resulting (but not necessarily delivered)
history $h_{(s+2)_j}$.
Because $j$ chooses some best history from set $R''_{s_j}$,
$h_{(s+2)_j}$ cannot have strictly lower priority than $h_{(s+2)_i}$,
otherwise $j$ would instead choose $h_{(s+2)_i}$.
So we can subsequently assume that the priority of $h_{(s+2)_j}$
is greater than or equal to that of $h_{(s+2)_i}$.

Every history occurring in $j$'s set $R''_{s_j}$, however,
is a proposal derived (via the round's second \broadcast call)
from a member of some set $B'_{s_k}$ 
that the first \broadcast call returned to some node $k$.
Because $B'_{s_k} \subseteq R'_{s_j}$ by \tsb's 
broadcast spread property,
both $h_{(s+2)_i}$ and $h_{(s+2)_j}$ must therefore also appear in $R'_{s_j}$.
But then $h_{(s+2)_i}$ cannot be uniquely best in $R'_{s_j}$,
satisfying the second condition on $i$ delivering $h_{(s+2)_i}$,
unless $h_{(s+2)_j} = h_{(s+2)_i}$.
\end{proof}

\lemliveness*
\begin{proof}
We will show that in the absence of a tie for best priority,
node $i$'s probability of successfully finalizing a round is ${t_b}/n$.
Since a round without a tie thus fails with probability at most $1-{t_b}/n$,
by the Union Bound,
the overall probability of round failure is at most $1-{t_b}/n+p_t$.

Let $N_{B'_{s_i}}$, $N_{R'_{s_i}}$, $N_{B''_{s_i}}$, $N_{R''_{s_i}}$ each be
the subsets of nodes $\{1,\dots,n\}$
whose messages $i$'s broadcast calls returned in its respective sets
$B'_{s_i}$, $R'_{s_i}$, $B''_{s_i}$, $R''_{s_i}$
(Definition~\ref{def:tsb}).
By the above independence assumption,
the network adversary's choices of these sets
does not depend on the content of messages or their priority values.

If the set $B''_{s_i}$ returned from $i$'s second broadcast
contains the round's unique globally-best history $\hat{h}_s$,
which exists due to our exclusion of ties above,
then $i$ will necessarily choose $\hat{h}_s$ and deliver it.
This is because $\hat{h}_s$ must also be in $R''_{s_i}$ and in $R'_{s_i}$,
and no other proposal exists in either set with priority
greater than or equal to that of $\hat{h}_s$.

This desirable event that $\hat{h}_s \in B''_{s_i}$ occurs
if at least one node $j \in N_{B''_{s_i}}$
chose $\hat{h}_s$ as its intermediate history $h''_{s_i}$
and broadcast it in $j$'s second call to \broadcast.
Since the probability of this event occuring for {\em at least one} node $j$
is no less than the probability of this event occuring
for {\em any specific} node $j \in N_{B''_{s_i}}$,
we now conservatively focus on analyzing this probability
of any specific such node $j \in N_{B''_{s_i}}$ choosing $\hat{h}_s$.

If the set $B'_{s_j}$ returned from $j$'s first broadcast
contains the round's unique globally-best history $\hat{h}_s$,
then $j$ will necessarily choose $h''_{s_j} = \hat{h}_s$
and broadcast it in $j$'s second \broadcast call.
This desirable event occurs in turn
if $N_{B'_{s_j}}$ includes the node $k$ that proposed
the unique globally-best history $\hat{h}_s$ in this round.
Since all nodes choose their priorities from the same random distribution,
each node has an equal chance of proposing
the globally-best history $\hat{h}_s$.
Since $|N_{B'_{s_j}}| \ge t_b$,
node $j$ therefore sees $\hat{h}_s$ in its set $B'_{s_j}$
with a probability of at least $t_b/n$.

Node $i$ therefore sees $\hat{h}_s$ in its set $B''_{s_j}$
and delivers a history in this round
with a probability of at least $t_b/n$.
\end{proof}

\thmqsc*
\begin{proof}
\emph{Liveness:}
\qsc regularly advances time forever on non-failing nodes
by calling \broadcast twice each time through an infinite loop,
at each step
delivering a history with some independent nonzero probability
(Lemma~\ref{lem:liveness}).
These delivered histories grow in length by one message
each time through the loop.
Therefore, if $h$ is the longest history delivered by time-step $s$
on a non-failing node $i$,
then with probability 1 there is eventually some future time-step $s' > s$
at which node $i$ delivers a longer history $h'$ ($|h'| > |h|$),
thereby satisfying liveness.

\emph{Validity:}
\xxx{ES: $\delta = 2$}
If \qsc invokes $\deliver(h' || p)$ at step $s'$ on node $j$,
then by the \tsb receive threshold property
$p$ is a proposal $\proposal{i,m,r}$
that some node $i$ appended to its internal history $h_i$
and broadcast at step $s = s'-2$,
at the beginning of the same \qsc round (main loop iteration).

\emph{Consistency:}
If \qsc delivers $h$ at step $s$ on node $i$,
then delivers $h'$ at step $s' \ge s$ on node $j$,
then by induction over $s'-s$,
using Lemma~\ref{lem:agreement} as the base case,
and using Lemma~\ref{lem:preservation} in the inductive step,
$h$ must be a prefix of $h'$.
\end{proof}

\subsection{Threshold Logical Clocks (\tlc)}
\label{sec:proofs:tlc}

This section contains correctness proofs
for the threshold logical clock algorithms
in Section~\ref{sec:tlc}.

\thmtlcr*
\begin{proof}
Provided \tlcr terminates,
it satisfies the \tsb's lock-step synchrony property (Definition~\ref{def:tsb}),
because $|\vec{R}|$ represents the current time-step at each invocation
counting from 2,
and each \tlcr call adds exactly one element to $\vec{R}$.
Because at most $f \le n - t_r$ nodes can fail,
each non-failed node eventually receives a threshold $t_r$
of messages from the $t_r$ non-failed nodes at each time-step,
ensuring that each \tlcr call eventually terminates
and successfully advances logical time.

\tlcr satisfies the \tsb receive threshold property by construction,
\ie, by not returning until it accumulates and returns a receive-set $R$
of size at least $t_r$ --
or until it obtains such a set $R$ all at once
by catching up to another node via a message from a future time-step.
Because messages are pairwise-ordered between nodes,
the condition $s' > |R|$ implies $s' = s+1$.
Because the internal receive-sets consist of pairs $\tlcmsg{j,m'}$
representing the sending node $j$ and message $m'$
that node $j$ broadcast in the same time-step,
the returned set $R$ contains messages sent by at least $t_r$ nodes
even if multiple nodes send the same message,
ensuring that the required node-set $N_R$ exists (Definition~\ref{def:tsb}).

\tlcr trivially satisfies
the broadcast threshold and broadcast spread properties
in Definition~\ref{def:tsb}
by always returning an empty broadcast set $B$,
thereby making no broadcast threshold promises to be fulfilled.
\end{proof}

\thmtlcbpartial*
\begin{proof}
In the second \tlcr call,
each node $i$ collects at least $t_r$ nodes' receive-sets
from the first \tlcr call,
each of which contains at least $t_r$ nodes' first-round messages.
We represent node $i$'s observations as a \emph{view matrix}
with $t_r$ rows (one per-receive set) and $n$ columns (one per node),
such that each cell $j,k$ contains 1
if $i$'s receive-set $j$
indicates receipt of node $k$'s message from the first \tlcr round,
and 0 otherwise.

Node $i$'s $t_r \times n$ view matrix
contains at least $t_r^2$ one bits,
and hence at most $t_r (n-t_r)$ zero bits.
To prevent $t_b$ nodes' messages from reaching at least $t_s$ nodes each
in $i$'s view,
the network must schedule the deliveries seen by $i$
so that at least $n-t_b+1$ columns of $i$'s view matrix
each fail to contain at least $t_s$ one bits.
Each such failing column
must contain at least $t_r-t_s+1$ zero bits.
Since there are at most $t_r (n-t_r)$ zero bits total,
there can be at most $f_b = t_r (n-t_r) / (t_r-t_s+1)$ failing columns.
The matrix must therefore have at least
$n-f_b$ non-failing columns
representing reliable broadcasts to at least $t_s$ nodes each.
\tlcb therefore satisfies the required broadcast threshold $t_b$
since $t_b \le n-f_b$.
\end{proof}

\thmtlcbfull*
\begin{proof}
By Theorem~\ref{thm:tlcb-partial},
the returned broadcast set $B$
contains the messages sent by at least $t_b$ nodes
in the first \tlcr step.
Consider any such message $m \in B$
and any node $i$ that completes this \tlcb step without failing.

By construction, node $i$'s set $B$ contains only messages
$i$ knows have been received by at least $t_s$ nodes.
Therefore,
there is some set $N_s \subseteq \{1,\dots,n\}$ of nodes
such that $|N_s| \ge t_s$,
and for each node $j \in N_s$,
the intermediate receive set $R'$ on node $j$ contains $m$.

Further,
due to the receive threshold $t_r$ enforced by \tlcr,
the message set $R''$ returned on node $i$
must contain the intermediate message sets $R'$
that were returned on at least $t_r$ nodes.
That is, there is some set $N_r \subseteq \{1,\dots,n\}$ of nodes
such that $|N_r| \ge t_r$,
and for each node $j \in N_r$,
the intermediate receive set $R'$ returned on node $j$
is a subset of $R''$ on node $i$.

Because $t_r + t_s > n$,
the sets $N_s$ and $N_r$ must therefore overlap by at least one node $k$.
Node $k$ therefore received message $m$ in its intermediate set $R'$,
and thus in turn must have passed $m$ on to $i$ via the second \tlcr step.
Therefore, message $m$ must be in the receive set finally
returned by \tlcb on node $i$.
Since this applies to all messages $m \in B$ and all nodes $i$,
\tlcb therefore implements $\tsb(t_r,t_b,n)$ full-spread synchronous broadcast.
\end{proof}

\thmtlcw*
\begin{proof}
\tlcw satisfies the \tsb's lock step synchrony because 
(a) each call to \tlcw only adds one element to $\vec{R}$, which represents the current time-step,
if it terminates, 
and (b) because at most $f \le n - t_b$ nodes can fail, each
non-failed node eventually receives at least $t_b$ messages from
the non-failed nodes guaranteeing that each call to \tlcw eventually 
terminates and advances the logical time.
\tlcw satisfies both the broadcast threshold and broadcast spread properties
by construction. Specifically, \tlcw does not return until it accumulates and returns a broadcast-set
$B$ of size at least $t_b$ or until it obtains such a set by catching up to another node via a message
from a future time step. The returned broadcast-set $B$ consists of at least $t_b$ fully witnessed messages $m'$ 
(satisfying broadcast threshold), where each node $j$ announces that
  $\tlcmsg{j,m'}$ has been fully witnessed only after 
its message $m'$ was acknowledged by $t_s$ nodes (satisfying broadcast spread), given that at most $f$ nodes can fail
and $t_b, t_s \leq n-f $. Finally, since $B \subseteq R$, we get $t_r \geq t_b$.
\end{proof}

\thmtlcf*
\begin{proof}
The proof is identical in essence to that of Theorem~\ref{thm:tlcb-full}.
\end{proof}

\section{\qsc model in Erlang}
\label{sec:erl}

To illustrate \qsc more concretely,
this section lists a full working model implementation of 
\qsc atop \tlcb and \tlcr in
\href{https://www.erlang.org}{Erlang}~\cite{armstrong13programming}.
The model implements nodes as Erlang processes interacting via message passing,
in less than 73 code lineas as counted by \texttt{cloc}~\cite{cloc}.
Of these, only 37 code lines comprise the consensus algorithm itself,
the rest representing test framework code.

Erlang is particularly well-suited to modeling \qsc,
being a distributed functional programming language with a concise syntax.
As a result, the actual working Erlang code
is not much longer in line count than the pseudocode
in Algorithms~\ref{alg:qsc}, \ref{alg:tlcr}, and~\ref{alg:tlcb}
that it implements.

Erlang's selective receive capability~\cite{vinoski12concurrency},
in particular,
simplifies implementation of \tlcr.
Selective receive allows \tlcr
to receive messages for the current time-step
and discard messages arriving late for past time-steps,
while saving messages arriving early for future time-steps
in the process's mailbox for later processing.

\arxiv{
\lstloadlanguages{Erlang}
}{}

\begin{tiny}

\subsection{\texttt{qsc.erl}: Erlang code listing}
\begin{lstlisting}[language=Erlang,columns=flexible]
-module(qsc).
-export([qsc/1, tests/0]).

% Node configuration is a tuple defined as a record.
-record(config, {node, tr, tb, ts, pids, choose, random, deliver}).

% A history is a record representing the most recent in a chain.
-record(hist, {step, node, msg, pri, pred}).

% qsc(C) -> (never returns)
% Implements Que Sera Co nsensus (QSC) atop TLCB and TLCR.
qsc(C) -> qsc(C, 1, #hist{step=0}).	% start at step 1 with placeholder pred
qsc(#config{node=I, choose=Ch, random=Rv, deliver=D} = C, S0, H0) ->
	H1 = #hist{step=S0, node=I, msg=Ch(C, S0), pri=Rv(), pred=H0},
	{S1, R1, B1} = tlcb(C, S0, H1),	% Try to broadcast (confirm) proposal
	{H2, _} = best(B1),		% Choose some best eligible proposal
	{S2, R2, B2} = tlcb(C, S1, H2),	% Re-broadcast it to reconfirm proposal
	{Hn, _} = best(R2),		% Choose best eligible for next round
	{H3, Unique} = best(R1),	% What is the best potential history?
	Final = lists:member(Hn, B2) and (Hn == H3) and Unique,
	if	Final -> D(C, S2, Hn), qsc(C, S2, Hn);	% Deliver history Hn
		true -> qsc(C, S2, Hn)	% Just proceed to next consensus round
	end.


% best(L) -> {B, U}
% Find and return the best (highest-priority) history B in a nonempty list L,
% and a flag U indicating whether B is uniquely best (highest priority) in L.
best([H]) -> {H, true};			% trivial singleton case
best(L) ->
	Compare = fun(#hist{pri=AR}, #hist{pri=BR}) -> AR >= BR end,
	[#hist{pri=BR} = B, #hist{pri=NR} | _] = lists:sort(Compare, L),
	{B, (BR /= NR)}.


% tlcb(C, S, H) -> {S, R, B}
% Implements the TLCB algorithm for full-spread synchronous broadcast.
tlcb(#config{ts=Ts} = C, S0, H) ->
	{S1, R1, _} = tlcr(C, S0, H),	% Step 1: broadcast history H
	{S2, R2, _} = tlcr(C, S1, R1),	% Step 2: re-broadcast list we received
	R = sets:to_list(sets:union([sets:from_list(L) || L <- [R1 | R2]])),
	B = [Hc || Hc <- R, count(R2, Hc) >= Ts],
	{S2, R, B}.			% New state, receive and broadcast sets

% count(LL, H) -> N
% Return N the number of lists in list-of-lists LL that include history H.
count(LL, H) ->
	length([L || L <- LL, lists:member(H, L)]).


% tlcr(C, S, M) -> {S, R, nil}
% Implements the TLCR algorithm for receive-threshold synchronous broadcast.
tlcr(#config{pids=Pids} = C, S, M) ->
	[P ! {S, M} || P <- Pids],		% broadcast next message
	tlcr_wait(C, S, []).			% wait for receive threshold
tlcr_wait(#config{tr=Tr} = C, S, R) when length(R) < Tr ->
	receive	{RS, RM} when RS == S -> tlcr_wait(C, S, [RM | R]);
		{RS, _} when RS < S -> tlcr_wait(C, S, R)	% drop old msgs
	end;	% when RS > S message stays in the inbox to be received later
tlcr_wait(_, S, R) -> {S+1, R, nil}.


% Run a test-case configured for a given number of potentially-failing nodes F,
% then signal Parent process when done.
test(F, Parent, Steps) ->
	% Generate a standard valid configuration from number of failures F.
	N = 3*F, Tr = 2*F, Tb = F, Ts = F+1,
	io:fwrite("Test N=~p F=~p~n", [N, F]),

	% Function to choose message for node I to propose at TLC time-step S.
	Choose = fun(#config{node=I}, S) -> {msg, S, I} end,

	% Choose a random value to attach to a proposal in time-step S.
	% This low-entropy random distribution is intended only for testing,
	% so as to ensure a significant rate of ties for best priority.
	% Production code should use high-entropy cryptographic randomness for
	% maximum efficiency and strength against intelligent DoS attackers.
	Random = fun() -> rand:uniform(N) end,

	% The nodes will "deliver" histories by sending them back to us.
	Tester = self(),		% Save our PID for nodes to send to
	Deliver = fun(C, S, H) -> Tester ! {S, C#config.node, H} end,

	% Receive a config record C and run QSC with that configuration.
	RunQSC = fun() -> receive C -> qsc(C) end end,

	% Launch a process representing each of the N nodes.
	Pids = [spawn(RunQSC) || _ <- lists:seq(1, N)],

	% Send each node its complete configuration record to get it started.
	C = #config{ tr = Tr, tb = Tb, ts = Ts, pids = Pids, 
		choose = Choose, random = Random, deliver = Deliver},
	[lists:nth(I, Pids) ! C#config{node=I} || I <- lists:seq(1, N)],

	% Wait until the test has completed a certain number of time-steps.
	test_wait(Parent, Pids, Steps, #hist{step=0}).

% Wait for a test to finish and consistency-check the results it commits
test_wait(Parent, Pids, Steps, Hp) ->
	receive	{S, I, H} when S < Steps ->
			io:fwrite("~p at ~p committed ~P~n", [I, S, H, 8]),
			test_wait(Parent, Pids, Steps, test_check(Hp, H));
		{_, _, _} ->
			[exit(P, kill) || P <- Pids],	% stop all our nodes
			Parent ! {}     		% signal test is done
	end.

% test_check(A, B) -> H
% Check two histories A and B for consistency, and return the longer one.
test_check(#hist{step=AC,pred=AP} = A, #hist{step=BC} = B) when AC > BC ->
	test_check(AP, B), A;		% compare shorter prefix of A with B
test_check(#hist{step=AC} = A, #hist{step=BC,pred=BP} = B) when BC > AC ->
	test_check(A, BP), B;		% compare A with shorter prefix of B
test_check(A, B) when A == B -> A;
test_check(A, B) -> io:fwrite("UNSAFE ~P /= ~P", [A, 8, B, 8]), A.

% Run QSC and TLC through a test suite.
tests() ->
	Self = self(),				% Save main process's PID
	Test = fun(F) -> 			% Function to run a test case
		Run = fun() -> test(F, Self, 1000) end,
		spawn(Run),			% Spawn a tester process
		receive {} -> {} end		% Wait until tester child done
	end,
	[Test(F) || F <- [1,2,3,4,5]],		% Test several configurations
	io:fwrite("Tests completed~n").

\end{lstlisting}

\end{tiny}

\section{\qscod: Client-driven on-demand \qsc with \tlcb}
\label{sec:odb}

\begin{algorithm*}[t!]
\caption{\qscod: client-driven execution of \qsc over \tlcb}
\SetCommentSty{textnormal}
\DontPrintSemicolon
\SetKwInOut{Config}{Configuration}
\SetKwInOut{State}{Global state}
\SetKwInOut{Sync}{Synchronization}
\SetKwInOut{Input}{Function input}
\SetKwInOut{Output}{Function output}
\SetKwFor{Forever}{forever}{}{end}

\Config{node number $i$ this thread drives, number of nodes $n$,
	thresholds $t_r \le n$ and $t_s \le n$}
\Config{functions \randval, $\kvwrite_i$, $\kvread_i$, \deliver}
\BlankLine

\State{client cache $C$ of servers' state;
	$C_{j,k}$ holds value stored on server $j$ under key $k$ if known}
\Sync{$\waitm() \rightarrow m$
	waits for and returns the next message this client wishes to commit}
\Sync{$\kvwait(k) \rightarrow R$
	waits to collect and return set $R$
	of cached values $C_{j,k}$
	from $\ge t_r$ nodes}

\BlankLine

$q \leftarrow 1$
	\tcp*{ficticious initial round}
$h \leftarrow \{\}$
	\tcp*{empty history for ficticious initial round}
$m \leftarrow$ \waitm() \tcp*{wait for first message this client
	wishes to commit}
\Forever(\tcp*[f]{loop until any client thread
		determines $m$ was committed}){}{
$q \leftarrow q+1$
	\tcp*{advance consensus round number}
$r \leftarrow \randval()$ \tcp*{choose proposal priority using private randomness}
$h' \leftarrow h_c \leftarrow \tlcmsg{\hash(h),m,r}$
	\tcp*{this client's preferred proposal for this round}

\BlankLine
$\kvwrite_i(\angles{q,1},\angles{h,h'})$
	\tcp*{try to record inputs to first \tlcb}
$\angles{h,h'} \leftarrow C_{i,\angles{q,1}}
		\leftarrow \kvread_i(\angles{q,1})$
	\tcp*{finalize our view of inputs to first \tlcb}
$R'_1 \leftarrow \kvwait(\angles{q,1})$
	\tcp*{tentative \tlcb initial broadcast outcome}

\BlankLine
$\kvwrite_i(\angles{q,2},R'_1))$
	\tcp*{try to record inputs to first \tlcb re-broadcast}
$R'_1 \leftarrow C_{i,\angles{q,2}} \leftarrow \kvread_i(\angles{q,2})$
	\tcp*{finalize our view of inputs to first \tlcb re-broadcast}
$R''_1 \leftarrow \kvwait(\angles{q,2})$
	\tcp*{tentative \tlcb re-broadcast outcome}
\BlankLine

$R_1 \leftarrow \bigcup(\{R'_1\} \cup R''_1)$
	\tcp*{tentative receive-set return from first \tlcb}
$B_1 \leftarrow \{m'\ |\ $at least $t_s$ messages-sets in $R''_1$ contain $m'\}$
	\tcp*{tentative broadcast-set return from first \tlcb}
$h'' \leftarrow$ any best history in $B_1$
	\tcp*{tentative history input to second \tlcb}

\BlankLine
$\kvwrite_i(\angles{q,3},\angles{R_1,B_1,h''})$
	\tcp*{try to record inputs to second \tlcb}
$\angles{R_1,B_1,h''} \leftarrow C_{i,\angles{q,3}}
		\leftarrow \kvread_i(\angles{q,3})$
	\tcp*{finalize our view of inputs to second \tlcb}
$R'_2 \leftarrow \kvwait(\angles{q,3})$
	\tcp*{tentative \tlcb initial broadcast outcome}

\BlankLine
$\kvwrite_i(\angles{q,4},R'_2)$
	\tcp*{try to record inputs to second \tlcb re-broadcast}
$R'_2 \leftarrow C_{i,\angles{q,4}} \leftarrow \kvread_i(\angles{q,4})$
	\tcp*{finalize our view of inputs to second \tlcb re-broadcast}
$R''_2 \leftarrow \kvwait(\angles{q,4})$
	\tcp*{tentative \tlcb re-broadcast outcome}

\BlankLine
$R_2 \leftarrow \bigcup(\{R'_2\} \cup R''_2)$
	\tcp*{tentative receive-set return from second \tlcb}
$B_2 \leftarrow \{m'\ |\ $at least $t_s$ messages-sets in $R''_2$ contain $m'\}$
	\tcp*{tentative broadcast-set return from second \tlcb}
$h \leftarrow$ any best history in $R_2$
	\tcp*{tentative history outcome for this round}
\If(\tcp*[f]{history $h$ has no competition}){$h = h_c$ and $h \in B_2$ and $h$ is uniquely best in $R_1$}{
	$\deliver(h)$
		\tcp*{deliver newly-committed history containing $m$}
	$m \leftarrow$ \waitm() \tcp*{wait for next message this client
		wishes to commit}
}
}

\label{alg:qscod}
\end{algorithm*}

To provide a concrete illustration
of the on-demand approach to implementing \qsc and \tlc
outlined in Section~\ref{sec:od},
Algorithm~\ref{alg:qscod} shows pseudocode for
client-driven \qsc built atop full-spread \tlcb
(Section~\ref{sec:tlcb:full}).

In \qscod,
each client node wishing to submit proposals and drive consensus
locally runs $n$ concurrent instances of Algorithm~\ref{alg:qscod},
typically in separate threads,
one for each of the $n$ servers providing key-value stores.
Initially and after each successful commitment of a client's proposal,
the client invokes \waitm to wait for the next message
to submit as a proposal. 
The client may be quiescent for arbitrarily long in \waitm,
during which the client produces no interaction with the servers
(but other clients can propose messages and drive consensus
in the meantime).

When \waitm returns the next message $m$ to be committed,
each client thread actively drives the key-value state
of its respective server forward --
in local cooperation with other client threads driving other servers --
to complete as many consensus rounds as necessary
to commit the client's proposed message $m$.

Since each \qsc round invokes \tlcb twice,
which in turn invokes \tlcr twice,
each consensus round requires four \tlcr time-steps.
We could model each \tlcr round as having $t_r+1$ state transitions:
one representing a given node $i$'s initial broadcast,
the rest for each of the $t_r$ messages subsequently ``received'' by $i$
as its condition to advance logical time.
It is possible and more efficient, however,
to summarize the effects of all simulated message ``receives'' in a time-step
as part of the \kvwrite representing the node's {\em next} broadcast.
With this approach, Algorithm~\ref{alg:qscod} requires
only four pairs of \kvwrite/\kvread requests
to each server per consensus round,
one pair for each of the four total \tlcr invocations.

Coordination between the simulated consensus nodes
occurs via the client's locally-shared cache $C$ of key-value pairs
that have been read from the $n$ servers so far.
After attempting to write a value to a key,
then reading back that key to learn what value was actually written
by the ``winning'' client,
each client thread invokes \kvwait
to wait until $t_r$ total threads
also write and read corresponding values for that key.
The client thread representing node $i$
takes this locally-determined set as a tentative, {\em possible}
receive-set of size $t_r$ for node $i$ --
but neither this nor anything computed from it
may be considered ``definite'' until
the next \kvwrite/\kvread pair commencing the next \tlcr time-step.

Each client-side instance of \qscod evaluates the \qsc finality conditions,
deciding whether the client's message has been successfully committed,
based on tentative information not yet finalized on the corresponding server.
This may seem like a problem, but is not.
Like any actual broadcast-based server implementation of \qsc,
a client thread will observe the finality conditions for history $h$
only when it is ``inevitable'' that all servers commit history $h$ --
regardless of whether or not they know that $h$ is committed.
A client thread might observe that $h$ is final,
deliver it to the application,
then lose a race to commit that result to the server
at the start of the next time-step --
but this means only that the simulated server
does not ``know'' that $h$ is committed,
even though the client in question (correctly) knows this fact.

\com{
\subsection{Complexity analysis}

If messages $m$ are of $O(1)$ size --
or if messages are $O(n)$ size but we use $O(1)$-size cryptographic hashes
to reference them in the message-sets written to the servers --
then all of the \kvwrite and \kvread values in \qscod
are likewise of $O(n)$ size.
Notice for example that \qscod never needs to communicate
a ``set of sets'' ...

Because \qscod computes all the relevant results
that depend on a particular time-step's receive-set,
then attempts to \kvwrite them to the server at once,
only the size of those computed results incur communication costs.

}

\xxx{ note: if a client fails,
leaving a \qsc consensus round only partially-completed,
the next client will automatically catch up to this partial state
and complete the consensus round.
If it is important not to leave consensus rounds partially-complete
during periods of quiescence,
then a garbage-collection processing acting like a client
can periodically finish any partially-completed consensus rounds.
}

\subsection{Model \qscod implementation in Go}
\label{sec:odb:go}

To illustrate the operation of \qscod more concretely,
this section finally presents a simple but fully-functional
model implementation of \qscod in the \href{https://golang.org}{Go language}.
The model implements nodes as goroutines communicating via shared memory
instead of via real network connections,
and is only 200 code lines as counted by \texttt{cloc}~\cite{cloc}
including test infrastructure (less than 125 lines without).
Despite its simplicity and limitations,
this model implements all the fundamental elements of \qscod,
and can operate in truly distributed fashion by filling in
core for remote access to key/value stores representing the consensus nodes,
including the marshaling and unmarshaling of stored values.
The latest version of this model may be found at
\url{https://github.com/dedis/tlc/tree/master/go/model/qscod}.

\long\def\gomodel{

\arxiv{
\lstloadlanguages{Go}
}{}

\begin{tiny}

\subsection{\texttt{cli.go}: \qscod client model}
\lstinputlisting[columns=flexible,tabsize=2]{src/go/cli.go}

\subsection{\texttt{cli\_test.go}: \qscod client testing framework}
\lstinputlisting[columns=flexible,tabsize=2]{src/go/cli_test.go}

\end{tiny}

}

\arxiv{
	\lstset{language=Go}
	\gomodel
}{
	\gomodel
}

\end{document}